\documentclass[11pt]{article}
\pdfoutput=1

\usepackage{fullpage}
\usepackage{algorithm}
\usepackage{algorithmic}

\usepackage[numbers]{natbib}
\usepackage{hyperref}
\usepackage{amsthm}
\usepackage{amsmath}
\usepackage{amsfonts}
\usepackage{breqn}
\usepackage{graphicx}
\usepackage{subfig}
\usepackage[table]{xcolor}
\usepackage{multirow}

\theoremstyle{definition}

\theoremstyle{remark}

\newtheorem{theorem}{Theorem}
\newtheorem{lemma}{Lemma}

\newtheorem{definition}{Definition}

\newcommand{\tabincell}[2]{\begin{tabular}{@{}#1@{}}#2\end{tabular}}

\begin{document}

\title{PrivMin: Differentially Private MinHash\\
for Jaccard Similarity Computation}
\author{Ziqi Yan\thanks{Beijing Key Laboratory of Security and Privacy in Intelligent Transportation, Beijing Jiaotong University, Beijing 100044, China. \texttt{\{zichiyen,jqliu,zhan,qiushuo\}@bjtu.edu.cn}.}
\and Jiqiang Liu\footnotemark[1]\ {~}
\and Gang Li\thanks{School of Information Technology, Deakin University, Geelong VIC 3125, Australia. \texttt{gang.li@deakin.edu.au}.}
\and Zhen Han\footnotemark[1]\ {~}
\and Shuo Qiu\footnotemark[1]\ {~}
}

\maketitle
\thispagestyle{empty}

\begin{abstract}
In many industrial applications of big data,
the \emph{Jaccard Similarity Computation} has been
widely used to measure the distance between two profiles
or sets respectively owned by two users.
Yet,
one semi-honest user with unpredictable knowledge may also
deduce the private or sensitive information (e.g., the existence
of a single element in the original sets)
of the other user via the shared similarity.

In this paper,
we aim at solving the privacy issues in Jaccard similarity
computation with strict differential privacy guarantees.
To achieve this,
we first define the \emph{Conditional $\epsilon$-DPSO},
a relaxed differential privacy definition
regarding set operations,
and prove that
the \emph{MinHash-based Jaccard Similarity Computation}
(\texttt{MH-JSC})
satisfies this definition.
Then for achieving strict differential privacy in \texttt{MH-JSC},
we propose the \texttt{PrivMin} algorithm,
which consists of two private operations:
1) the \emph{Private MinHash Value Generation} that
works by introducing the \emph{Exponential} noise
to the generation of MinHash signature.
2) the \emph{Randomized MinHashing Steps Selection} that
works by adopting \emph{Randomized Response} technique
to privately select several steps within the MinHashing phase
that are deployed with the \emph{Exponential} mechanism.
Experiments on real datasets demonstrate that the proposed
\texttt{PrivMin} algorithm can successfully
retain the utility of the computed similarity
while preserving privacy.
\end{abstract}

\newpage
\thispagestyle{empty}

\tableofcontents

\newpage
\pagenumbering{arabic}

\section{Introduction}\label{sec-introduction}

With the widespread of real-world big data
applications such as recommendation systems
and social network,
\emph{similarity computation} has become
one essential process as it measures the distances
between different user profiles.
Utilizing the similarity between users or items,
the service providers can further carry out
data analytic tasks such as
clustering,
classification or recommendation.
Among the varieties of similarity measures,
\emph{Jaccard Similarity} is a popular one
that has been widely used to compare the similarity of two given sets.
More specifically,
for two sets $S_{A}$ and $S_{B}$,
their Jaccard similarity is defined as the ratio between
the size of their intersection and
the size of their union.

However,
because of the adoption of Jaccard similarity in real-world applications,
one increasing concern is the potential privacy leakage.
Let us consider one example scenario as described below.

\begin{description}\label{exp:attack}
\item[Example 1]
In cloud services,
the similarity computation may be available for users
who want to know the ``semantic distance'' between their data.

Assume a cloud platform which provides such a
service of Jaccard similarity computation
between two sets
$S_{A}$ and $S_{B}$ privately owned by users
$U_{A}$ and $U_{B}$,
respectively.
Each set contains a fixed number of textual tags
that reflect the user's reading preferences,
such as
$S_{A} =\{$\texttt{History}, \texttt{Politics}, \texttt{Science}, \texttt{Law}, \texttt{Travel}$\}$
and $S_{B} = \{$\texttt{History}, \texttt{Science}, \texttt{Travel}, \texttt{Cookbooks}, \texttt{Fiction}$\}$.
The cloud service may estimate the Jaccard similarity as
$\frac{|S_{A} \cap S_{B}|}{|S_{A} \cup S_{B}|} = 0.429$,
and make it available for users $U_{A}$ and $U_{B}$.

Remarkably,
based on the above shared similarity and
the prior knowledge that
the value of $|S_{A} \cup S_{B}|$
should fall into the range of $5$ to $10$,
$U_{B}$ can easily work out
the values of $|S_{A} \cap S_{B}|$ and $|S_{A} \cup S_{B}|$,
which are $3$ and $7$ respectively.
Moreover,
if $U_{B}$ further knows in advance that
$U_{A}$ does not like the book genres such as \texttt{cookbooks} and \texttt{fiction},
he would basically make sure
that his three tags in common with $U_{A}$
are $\{$\texttt{History}, \texttt{Science}, \texttt{Travel}$\}$.
In addition,
when we take into account the fact that there are two collusive users
$U_{B}$ and $U_{C}$
who are interested in $U_{A}$'s private information,
it would not take these collusive users
much background knowledge to achieve their purpose.
$U_{B}$  would also achieve the attack goal easily
through the similarity
with $U_{A}$ via different carefully constructed sets.
\end{description}

As shown in the above example,
users with background knowledge
can induce other users' private information with high probability
by observing their shared similarity.
Hence,
how to preserve the privacy in the Jaccard similarity computation
is an emerging issue that needs to be addressed.

In the past decade,
\emph{Differential Privacy} has emerged as a solid privacy model
with a provable privacy guarantee,
regardless of the adversary's background knowledge.
Recently,
some researches have focused on the privacy issue
in similarity computation by incorporating the \emph{differential privacy} mechanism.
Alaggan~{et~al.}~\cite{AlagganGK11OPODIS} proposed several
secure protocols to compute differentially private values
of \texttt{Scalar Product} and \texttt{Cosine} similarity.
Their follow-up paper~\cite{AlagganGK12SSS} proposed
a differentially private method for randomizing
the intermediate outputs instead of adding noise to the final \texttt{Cosine} similarity outputs.
Wong~{et~al.}~\cite{WongH14SWJ} first proposed a secure protocol
for a specific Jaccard similarity computation for the binary data.
However,
those tailored Jaccard similarity computations cannot
be generalized to other situations.
To the best of our knowledge,
there is limited researches that have addressed the privacy concerns
in the general Jaccard similarity computation
while maintaining the acceptable utility and efficiency.

As the advances in \emph{Hashing} techniques,
such as the \texttt{MinHash} and \texttt{SimHash},
the current research barriers can be tackled in a natural way.
The \emph{MinHash} technique~\cite{Broder1997IEEE} was proposed
to efficiently approximate the value of Jaccard similarity instead of the precise one,
so it can significantly improve the computation efficiency
when a large collection of data involved~\cite{BroderCFM98STOC, DasDGR07WWW}.
For convenience,
we refer to this processing workflow as
\emph{MinHash-based Jaccard Similarity Computation}
(\texttt{MH-JSC}).

In this paper,
we will present an intuition that
the \texttt{MH-JSC} is internally connected with a relaxed differential privacy,
because its expected error $\theta$ can be regarded as noise.
This intuition opens the opportunity to design a differentially private
Jaccard similarity computation algorithm,
which protects the certainty of presence/absence of any element in the original profile.
However,
there are still two main challenges
when designing the differentially private Jaccard similarity computation algorithm:
\begin{itemize}
\item
The first challenge is
how to measure the randomness within the \texttt{MH-JSC}
for further analyzing its relationship with the differential privacy.

\item
The second challenge lies on
how to leverage the minimum hash value computation process
within the MinHashing phase
for achieving strict differential privacy in \texttt{MH-JSC}
while maintaining an acceptable utility.
\end{itemize}

For the first challenge,
we investigate the relationship
between the \texttt{MH-JSC} and the differential privacy
via a relaxed differential privacy definition,
\emph{Conditional $\epsilon$-DPSO}.
Based on this,
we intend to design a differentially private Jaccard similarity
computation algorithm via the \emph{Exponential} mechanism,
which leverages the minimum hash value computation process
within the MinHashing phase.
As the introduced \emph{Exponential} noise will distort the utility in a large extent,
the second challenge can be solved
by introducing the \emph{Randomized Response} technique
to privately select some MinHashing steps
for adopting the \emph{Exponential} mechanism.

Based on these,
we finally present the \texttt{PrivMin} algorithm
to achieve the differentially private Jaccard similarity computation,
and the contributions in this paper can be summarized as follows:
\begin{itemize}
\item
Firstly,
through the relaxed differential privacy,
\emph{Conditional $\epsilon$-DPSO},
we investigate the randomness within the \texttt{MH-JSC}
and provide relevant privacy analysis in detail.

\item
Secondly,
we design a practical differential private
Jaccard similarity computation algorithm,
\texttt{PrivMin},
which maintains an acceptable utility.
Theoretical analysis and extensive experiments
are provided to verify the improved performance.
\end{itemize}

The rest of this paper is organised as follows.
We present the preliminaries and related works in
Section~\ref{sec-preliminaries},
and provide the problem statement in Section~\ref{sec-problem_statement}.
In Section~\ref{sec-dpso},
we define the \emph{Conditional $\epsilon$-DPSO}
to depict a relax situation
when considering differential privacy for set operations,
followed by theoretical privacy analysis of
the \texttt{MH-JSC} under this definition.
In Section~\ref{sec-private-jaccard-similarity-computation},
we describe the \texttt{PrivMin} algorithm
for achieving the differentially private MinHash-based
Jaccard similarity computation.
The theoretical privacy analysis and utility analysis of
the algorithm are proposed in Section~\ref{sec-analysis}.
Section~\ref{sec-experiment} presents experimental results,
and conclusions are given in Section~\ref{sec-conclusion}.

\section{Preliminaries and Related Works}
\label{sec-preliminaries}

This section reviews four fundamental concepts:
\emph{Jaccard Similarity},
\emph{MinHash},
\emph{Differential Privacy}
and \emph{Randomized Response},
and then briefly surveys the related works
in \emph{Differentially Private Similarity Computation}.

Table~\ref{tab:notations} lists the relevant notations
used in this paper.

\begin{table}[h]  \centering
\small
  \caption{Notations}\label{tab:notations}
  \begin{tabular}{l|l}
\hline
    \rowcolor[HTML]{EFEFEF}
    \textbf{Symbol}
    & \textbf{Description}\\
\hline
    $\texttt{MH-JSC}$
    & abbreviation of MinHash-based Jaccard similarity computation  \\ \hline

    $S$
    & user's private profile $S=\{s_{1},s_{2},...,s_{N}\}$  \\ \hline

    $J(S_{A},S_{B})$
    & original Jaccard similarity of $S_{A}$ and $S_{B}$  \\ \hline

    $\sigma$
    & a conventional notation to represent the value of $J(S_{A},S_{B})$  \\ \hline

    $J_{mh}(S_{A},S_{B})$
    & original MinHash-based Jaccard similarity of $S_{A}$ and $S_{B}$  \\ \hline

    $\theta$
    & expected error in $J_{mh}(S_{A},S_{B})$ compared with $J(S_{A},S_{B})$  \\ \hline

    $\widetilde{J_{mh}}(S_{A},S_{B})$
    & perturbed MinHash-based Jaccard similarity of $S_{A}$ and $S_{B}$  \\ \hline

    $\Delta J_{mh}$
    & the sensitivity of $\texttt{MH-JSC}$  \\ \hline

    $K$
    & number of hash functions  \\ \hline

    $h_{k}(S)$
    & hash values set of profile $S$ when given a hash function $h_{k}{(\cdot)}$   \\ \hline

    $min\{h_{k}(S)\}$
    & the minimum hash value in $h_{k}(S)$  \\ \hline

    $h_{(K)}(S)$
    & original MinHash signature vector  \\ \hline

    $\widetilde{h_{(K)}}(S)$
    & perturbed MinHash signature vector  \\ \hline

    $\overrightarrow{V}$
    & original flip vector  \\ \hline

    $P_{r}$
    & bit flipping probability in original flip vector generation  \\ \hline

    $\overrightarrow{V'}$
    & perturbed flip vector  \\ \hline

    $P_{t}$
    & bit flipping probability in perturbed flip vector generation  \\ \hline

    $\epsilon$
    & overall privacy budget  \\ \hline

\hline
\end{tabular}
\end{table}

\subsection{Preliminaries}

\begin{definition}[Jaccard Similarity]
Assume $S_{A}$ and $S_{B}$ are two sets owned by
user $A$ and user $B$ respectively.
Their Jaccard similarity is defined as
\begin{equation}\label{equ:Jsim}
J(S_{A},S_{B})=\frac{|S_{A} \cap S_{B}|}{|S_{A} \cup S_{B}|}.
\end{equation}
\end{definition}

\subsubsection{MinHash}\label{sec-minhash}

The MinHash was initially proposed
in~\cite{Broder1997IEEE,BroderCFM98STOC}
for quickly estimating
the similarity between two textual documents
which have been respectively expressed as
sets $S_{A}$ and $S_{B}$.
The basic intuition for the MinHash technique is
the replacement of the original sets $S_{A}$ and $S_{B}$
by their relevant \emph{MinHash Signatures}
$h_{(K)}(S_{A})$ and $h_{(K)}(S_{B})$
when computing the Jaccard similarity.
For the convenience of the following descriptions,
we refer to the above similarity estimating process as
\emph{MinHash-based Jaccard Similarity Computation} (\texttt{MH-JSC}).

In practice,
the \texttt{MH-JSC} between textual documents
usually involves three main phases:
the \emph{Shingling} phase to formulate
the textual documents into set representations,
the \emph{MinHashing} phase to generate
the relevant MinHash signatures,
followed by the approximate computation phase.

Specifically,
in the \emph{Shingling} phase,
the document is firstly segmented into $N$ parts (shingles)
and represented as the set
$S=\{s_{1},s_{2},...,s_{N}\}$;
in the process of the MinHashing phase,
$K$ hash functions $h_k$ with $k \in [1,K]$
are orderly applied to $S$ and generate
$h_{k}(S)=\{h_{k}(s_1),h_{k}(s_2),...,h_{k}(s_N)\}$,
and then the minimum hash value
$min\{h_{k}(S)\}$
is selected as the $k$-th element of
the MinHash Signature $h_{(K)}(S)$,
as shown in Fig~\ref{fig:minhashing}.

\begin{figure}[htbp]
\centering
\includegraphics[scale=0.5]{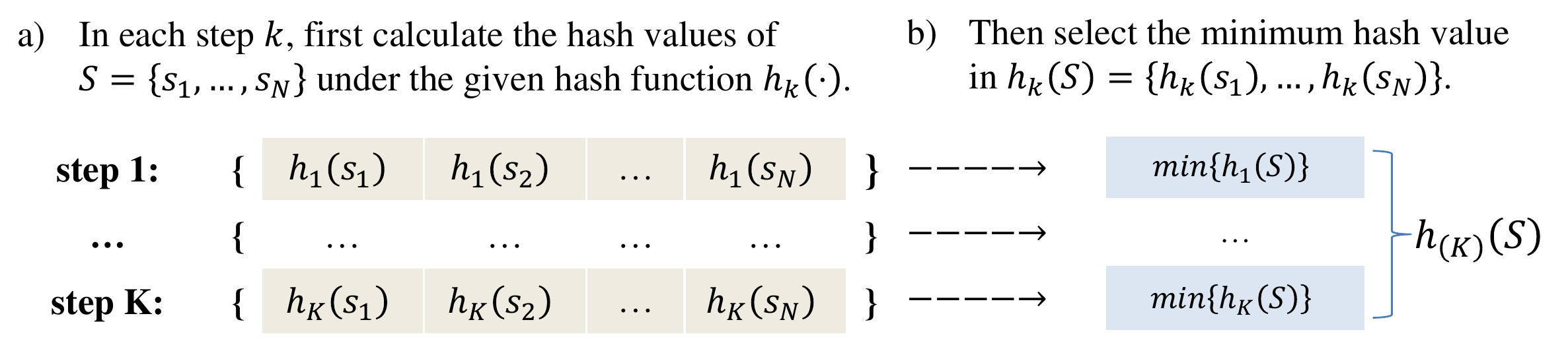}
\caption{MinHashing Phase}
\label{fig:minhashing}
\end{figure}

Given the MinHash signatures $h_{(K)}(S_{A})$ and $h_{(K)}(S_{B})$
for two textual documents $S_{A}$ and $S_{B}$,
an unbiased estimate of the Jaccard similarity between
$S_{A}$ and $S_{B}$ is formulated as
\begin{equation}\label{equ:Jsimmh}
J_{mh}(S_{A},S_{B})=\frac{|h_{(K)}(S_{A})
\cap h_{(K)}(S_{B})|}{K},
\end{equation}
with an expected error
$\theta = O(\frac{1}{\sqrt{K}})$.

For the convenience of description,
if we adopt the notation $\sigma$ to represent the value of $J(S_{A},S_{B})$,
the probability for $J_{mh}(S_{A},S_{B})$ to
fall into the range $[\sigma-\theta, \sigma+\theta]$
can be calculated via the following equation~\cite{Broder1997IEEE}:
\begin{equation}\label{equ:approx.prob}
p(K,\sigma,\theta)=
\sum_{K(\sigma-\theta) \leq t \leq K(\sigma+\theta)}
\binom{K}{t}(\sigma)^{t}(1-\sigma)^{K-t}.
\end{equation}

\subsubsection{Differential Privacy}

Differential privacy
is based on the principle that the output of a computation
should not allow inference about any element's presence
or absence from the computation's input.
Hence in the context of Jaccard similarity computation,
the present or absent status of the elements within input data
is expected to be protected under the rigorous
differential privacy definition which is described below.

\begin{definition}[$\epsilon$-Differential Privacy~\cite{Dwork11Encyclopedia}]
\label{def:DP}
A randomized algorithm $\mathcal{M}$ gives
$\epsilon$-differential privacy
if for all neighbour sets $S$ and $S'$ differing on at most one element,
and all $O \subseteq Range(\mathcal{M})$,
we have
\begin{center}
$Pr[\mathcal{M}(S) \in O] \leq e^{\epsilon}
\cdot Pr[\mathcal{M}(S') \in O]$.
\end{center}
\end{definition}

Algorithm $\mathcal{M}$ is associated with the \emph{sensitivity},
which measures the maximum change on the result of query function
$f$ when one element from the set $S$ changes~\cite{Dwork11CACM}.

\begin{definition}[Sensitivity]\label{def:sensitivity}
For any function $f: S \rightarrow \mathbb{R}^{d}$,
and for all $S$,
$S'$ differing in at most one element,
the sensitivity of $f$ is
$\bigtriangleup f=\max \limits_{S,S'} \Vert{f(S)-f(S')}\Vert_{1}$.
\end{definition}

To satisfy the definition of differential privacy,
two basic mechanisms are usually utilized:
the \emph{Laplace} mechanism and the \emph{Exponential} mechanism.
And the \emph{Laplace} mechanism is suitable for numeric output
and relies on the strategy of adding the \emph{Laplacian} noise
$Laplace(\cdot)$ to the query result~\cite{DworkMNS06TCC}.
It is formally defined as the following:

\begin{definition}[Laplace Mechanism]\label{def:lap-mech}
Given a function $f: S\rightarrow \mathbb{R}^{d}$,
the mechanism,
\begin{center}
$\mathcal{M}(S)=f(S)+(Y_{1},...,Y_{d})$,
\end{center}
where $Y_{i}$ are i.i.d random variables drawn from
$Laplace(\frac{\bigtriangleup f}{\epsilon})$.
\end{definition}

The \emph{Exponential} mechanism focuses on queries with
non-numeric output~\cite{McSherryT07FOCS}.
It pairs with an application dependent \emph{Score Function} $q(S,\psi)$,
which represents how good an output scheme $\psi$ is for query $q$.
The \emph{Exponential} mechanism is formally defined as

\begin{definition}[Exponential Mechanism]\label{def:exp-mech}
An \emph{Exponential} mechanism $\mathcal{M}$ is
$\epsilon$-differential privacy if
\begin{center}
$\mathcal{M}(S)=\{$return $\psi$ with the probability
$\propto \textsf{exp}(\frac{\epsilon q(S,\psi)}{2\bigtriangleup q})\}$.
\end{center}
\end{definition}

To guarantee the overall privacy
when it comes to a sequence of differentially private operations,
we have the following composition properties~\cite{McSherry09SIGMOD}.

\begin{theorem}[Sequential Composition]\label{thr:sequential-composition}
Given $n$ independent randomized algorithms
$\mathcal{A}_{1},\mathcal{A}_{2},...,\mathcal{A}_{n}$
where
$\mathcal{A}_{i} (1 \leq i  \leq n)$ satisfies
$\epsilon_{i}$-differential privacy,
a sequence of $\mathcal{A}_{i}$ over the dataset
$S$ satisfies $\epsilon$-differential privacy,
where $\epsilon = \sum \nolimits_{1}^{n}(\epsilon_{i})$.
\end{theorem}

\begin{theorem}[Parallel Composition]\label{thr:parallel-composition}
Given $n$ independent randomized algorithms
$\mathcal{A}_{1},\mathcal{A}_{2},...,\mathcal{A}_{n}$ where
$\mathcal{A}_{i} (1 \leq i  \leq n)$
satisfies $\epsilon_{i}$-differential privacy,
a sequence of $\mathcal{A}_{i}$
over a set of disjoint datasets $S_{i}$
satisfies $max(\epsilon_{i})$-differential privacy.
\end{theorem}

\subsubsection{Randomized Response}\label{sec-rr}

\emph{Randomized Response} is a commonly used survey technique
in statistics~\cite{Warner65}.
When a respondent is asked a sensitive question
for which the answer can be either \texttt{yes} or \texttt{no},
he has the opportunity to answer the question with plausible deniability.
To do so,
the respondent flips a biased coin before answering the question.
If the coin turns head with a probability $p$,
he gives his true answer;
otherwise,
he reports the opposite of the true answer.
It has pointed out that
\emph{Randomized Response} can
be regarded as a specific randomized algorithm
that satisfies the $\epsilon$-differential privacy,
if the coin flipping probability $p$ of the algorithm has
the following relationship with the privacy budget $\epsilon$
~\cite{ErlingssonPK14CCS,BassilyS15STOC}:

\begin{equation}
p = \frac{e^{\epsilon}}{1 + e^{\epsilon}}.
\end{equation}

\subsection{Related Works}

\subsubsection{Differentially Private Similarity Computation}

For the applications
such as recommender system~\cite{BartheDGKB13CSF, WongH14SWJ},
several works have been proposed to address
the potential privacy issues in two-party profiles computation~
\cite{AlagganGK11OPODIS, BoutetFGJK16Springer},
in user profiles collection~\cite{ShenCJ16ECML, ShenJ16CCS}
and in the profile related data releasing~\cite{AlagganGK12SSS, ZhuLZXY14PAKDD, ZhuLZXY16KAIS}
by the third party.

Most of these works were focused on
the distributed environments in which
the involved users are semi-honest while
the third party (if it existed) is assumed as semi-trusted or even untrusted.
Therefore,
users profile must be perturbed or encrypted before
being sent to other users or the third party
for further processing such as similarity computation,
an essential component in collaborative filtering.
In addition,
sometimes the released similarity also needs to be perturbed.

For achieving the differentially private similarity computation
by output perturbation,
the line of research was pioneered by Alaggan~{et~al.}~\cite{AlagganGK11OPODIS},
who introduced the \emph{Laplace} mechanism into
the \emph{Scalar Product} and \emph{Cosine} similarity computation.
In particular,
their proposed secure protocols were partially based on
\emph{Homomorphic Encryption} and worked by
adding the \emph{Laplacian} noise to the similarity.
Following a similar strategy,
Wong~{et~al.}~\cite{WongH14SWJ} presented a secure protocol
for a specific Jaccard similarity computation of binary data.

For profile perturbation,
Alaggan~{et~al.}~\cite{AlagganGK12SSS} considered
the scenario of profile release and proposed the BLIP mechanism
in which the Bloom filter of user profile
would be distorted by \emph{Randomized Response}
before being released to the public.
The \emph{Scalar Product} and \emph{Cosine} similarity
were considered in this work.
Barthe~{et~al.}~\cite{BartheDGKB13CSF} proposed
a two-party protocol for computing Hamming distance
between bit-vectors via \emph{Homomorphic Encryption} and the \emph{Laplace} mechanism.
Boutet~{et~al.}~\cite{BoutetFGJK16Springer} designed
an obfuscation protocol and a randomized dissemination protocol
for two-party Jaccard similarity computation.
Besides,
existing works focused also on
the perturbation of user profiles for
dataset release and for multi-level privacy needs
instead of specific similarity computation needs.
Zhu~{et~al.}~\cite{ZhuLZXY14PAKDD, ZhuLZXY16KAIS}
considered the privacy issues in releasing
and sharing of tagging datasets
in tagging recommender systems and
presented the private tagging release algorithm
\emph{PriTop} based on the topic generation model,
on the \emph{Laplace} mechanism and on the \emph{Exponential} mechanism.
Shen~{et~al.}~\cite{ShenCJ16ECML, ShenJ16CCS}
aimed to achieve multi-level privacy control
in user profile perturbation and
proposed the \emph{DP-MultiUPP} and \emph{EpicRec} frameworks
based on the \emph{Laplace} mechanism and optimization techniques.

Table~\ref{tab:dpsc-comparison} gives the comparison
among the existing works
for differentially private similarity computation.
The main details of our proposed \texttt{PrivMin} algorithm
are also listed in the table,
and the differences between our work and
the existing works will be discussed in the next section.

\begin{table}[h]  \centering
\scriptsize
  \caption{Differentially Private Similarity Computation Comparison}\label{tab:dpsc-comparison}
  \begin{tabular}{l|c|c|c|c}
\hline
    \rowcolor[HTML]{EFEFEF}
    \textbf{Related Work} & \textbf{Third Party Setting} & \textbf{Similarity Type} & \textbf{Perturbation Approach} & \textbf{Involved Method}\\
\hline
    Alaggan~{et~al.}~\cite{AlagganGK11OPODIS}
    & \tabincell{c}{None or \\Semi-trusted}
    & \tabincell{c}{Scalar Product, \\Cosine Similarity}
    & Output Perturbation
    & \tabincell{c}{Homomorphic Encryption, \\Laplace Mechanism}  \\ \hline

    Wong~{et~al.}~\cite{WongH14SWJ}
    & Semi-trusted
    & Jaccard Similarity
    & Output Perturbation
    & \tabincell{c}{Homomorphic Encryption, \\Laplace Mechanism}  \\ \hline

    Barthe~{et~al.}~\cite{BartheDGKB13CSF}
    & None
    & Hamming Distance
    & Profile Perturbation
    & \tabincell{c}{Homomorphic Encryption, \\Laplace Mechanism}  \\ \hline

    Alaggan~{et~al.}~\cite{AlagganGK12SSS}
    & None
    & \tabincell{c}{Scalar Product, \\Cosine Similarity}
    & Profile Perturbation
    & \tabincell{c}{Bloom Filter, \\Randomized Response}  \\ \hline

    Boutet~{et~al.}~\cite{BoutetFGJK16Springer}
    & None
    & Jaccard Similarity
    & Profile Perturbation
    & \tabincell{c}{Compact profile construction, \\Randomized Response}  \\ \hline

    \texttt{PrivMin}
    & Trusted
    & Jaccard Similarity
    & \tabincell{c}{Profile Perturbation}
    & \tabincell{c}{MinHash Signature, \\Exponential Mechanism, \\Randomized Response}  \\ \hline
\hline
\end{tabular}
\end{table}

\subsubsection{Discussion}

Based on the above analysis,
the differences between our work and the existing works
can be concluded in three aspects:
\begin{itemize}
  \item
Firstly,
compared to the distributed setting of the existing works,
our work is focused on the centralized setting.
Moreover,
this work mainly assumes that the third party is trusted
while the existing works generally assumed
a semi-trusted third party or no third party at all.
The main reason of such assumptions in our work is
that in many real-world applications
of recommendation and plagiarism detection,
the service providers always have access to users profiles
and then use their storage capacity and computing ability
to provide users with varieties of services.
Even so,
the proposed \texttt{PrivMin} algorithm can also be
extended for the untrusted third party scenario,
since it has the ability to release perturbed users profiles
before entering the similarity computation phase.

\item
Secondly,
for the research of
differentially private Jaccard similarity computation,
the existing work such as~\cite{WongH14SWJ} was
partially focused on applying the \emph{Laplace} mechanism to
the original Jaccard similarity computation equation,
which cannot maintain a high utility of the released similarity.
Besides,
due to its assumption of binary input data,
the current work failed to meet the privacy needs
in the application scenario as shown in Example~\ref{exp:attack}.
Moreover,
few attention has been devoted to the relationship
between MinHash-based Jaccard similarity computation and
differential privacy,
which is the basic rationale and main contribution of our work.

\item
Thirdly,
in order to adopt the \emph{Randomized Response} technique,
the existing works (e.g.,~\cite{AlagganGK12SSS, BoutetFGJK16Springer})
were focused on directly distorting user profiles
which are represented as binary expressions.
However,
if relying on some specific value computation process such as MinHashing,
we find that the combination of
\emph{Randomized Response} technique
with \emph{Exponential} mechanism could provide the possibility
to design a differentially private algorithm
with acceptable utility.
To the best of our knowledge,
this is the first attempt to
incorporate the \emph{Exponential} mechanism with
the \emph{Randomized Response} in the context of
differentially private Jaccard similarity computation.
\end{itemize}

\subsection{Summary}

For differentially private similarity computation,
the existing works have established
two perturbation strategies to address relevant issues,
and provided referential ideas and solutions
for differentially private algorithm design.
However,
currently there has been limited research attention in
\emph{the MinHash-based Jaccard Similarity Computation}
(\texttt{MH-JSC})
to design the differentially private
Jaccard similarity computation algorithm.
Based on the observation that
the internal randomness within \texttt{MH-JSC}
is related with a relaxed differential privacy,
this paper aims to address the following specific research issues:
\begin{itemize}
\item
How to measure the randomness within \texttt{MH-JSC}?

\item
How to achieve strict differential privacy in \texttt{MH-JSC}?
\end{itemize}

\section{Problem Statement}\label{sec-problem_statement}

This section first introduces the \emph{system and threat model}
considered in this work,
and then clearly presents
the differentially private Jaccard similarity computation problem,
along with its challenges.

\subsection{System and Threat Model}

Since in many real-life scenarios
users are expected to provide  their true data to the cloud platform
in order to access add-on services
such as accurate recommendation,
we assume that
the cloud platform will not collude with any user and is trusted.
Meanwhile,
the platform users are supposed to be semi-honest,
namely
they are willing to provide their own data to the cloud
but also curious about other users' sensitive information.

Fig.~\ref{fig:sysmodel} shows the system and treat model.

\begin{figure}[htbp]
\centering
\includegraphics[scale=0.5]{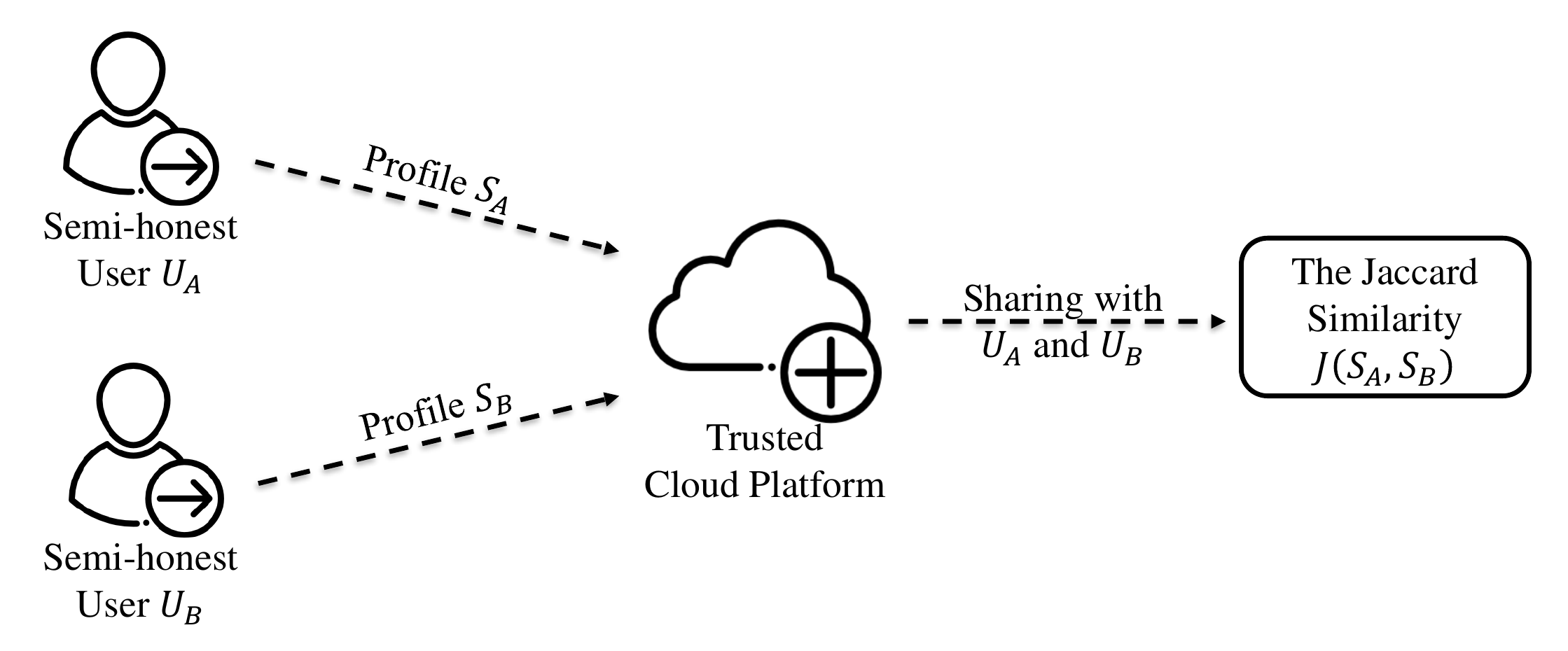}
\caption{System and Threat Model}
\label{fig:sysmodel}
\end{figure}

\subsection{Problem Definition}

In this work,
we are addressing the problem of
\emph{differentially private Jaccard similarity computation},
which can be described as follows.

Assume two users $U_{A}$ and $U_{B}$,
each of them respectively maintains the profile $S_{A}$ and $S_{B}$
on the trusted cloud platform.
Given their profiles $S_{A}$ and $S_{B}$ and
the privacy budget $\epsilon$,
the cloud platform is expected to
calculate the Jaccard similarity $J(S_{A},S_{B})$
and shares a perturbed version with these two users.
On one hand,
the overall similarity computation mechanism should
satisfy $\epsilon$-differential privacy
such that no semi-honest users can infer
the present or absent status of the elements
in other user profiles based on the shared similarity.
On the other hand,
the shared perturbed similarity should also maintain
acceptable utility for further data analysis or value-adding services.

\subsection{Research Issues and Challenges}

In this paper,
we aim to solve
the differentially private Jaccard similarity
computation problem
by leveraging the MinHash and \texttt{MH-JSC}.
However,
directly introducing differential privacy into the \texttt{MH-JSC}
brings up two major challenges.
\begin{description}
\item[How to measure the randomness within \texttt{MH-JSC}?]
As introduced in Section~\ref{sec-preliminaries},
the MinHash-based Jaccard Similarity Computation
can estimate the Jaccard Similarity
with an expected error $\theta$.
It seems that this error can be regarded as
a kind of internal randomized noise,
which makes it possible for \texttt{MH-JSC} to achieve differential privacy.
If the above hypothesis is proved right,
it is not necessary to add any extra external noise to \texttt{MH-JSC}
since the internal noise could be enough.

In Section~\ref{sec-dpso},
we will define $\epsilon$-DPSO and
conditional $\epsilon$-DPSO,
to analyze the relationship between
the randomness within the \texttt{MH-JSC} and differential privacy.

\item[How to achieve strict differential privacy in \texttt{MH-JSC}?]
Based on the randomness analysis
in Section~\ref{sec-dpso},
we will show that \texttt{MH-JSC} only satisfies
a relaxation of strict differential privacy,
the conditional $\epsilon$-DPSO.
For achieving strict differential privacy,
although the \emph{Laplace} mechanism can be applied to
perturbed the original MinHash-based Jaccard similarity,
the final utility of similarity will be distorted in a large extent.
In Section~\ref{sec-private-jaccard-similarity-computation}
we will adopt the \emph{Profile Perturbation} approach and
propose two private operations
to constitute the \texttt{PrivMin} algorithm,
along with the relevant privacy analysis
in Section~\ref{sec-analysis}.
This algorithm also exploits
the minimum hash value computation process
within the MinHashing phase.
In the meanwhile,
the \emph{Exponential} mechanism and \emph{Randomized Response}
will be carefully adopted for maintaining an acceptable utility.
\end{description}

\section{Randomness Analysis within MH-JSC}\label{sec-dpso}

In this section,
in order to study the relationship between the internal randomness
within the \texttt{MH-JSC} and the differential privacy,
we first provide a relaxed definition of
differential privacy ($\epsilon$-DPSO) for set operations,
and then prove that
the \texttt{MH-JSC} satisfies the $\epsilon$-DPSO.

\subsection{Differentially Private Set Operations}

In the definition of differential privacy~\ref{def:DP}
the neighbouring datasets are $S$ and $S'$
which differ in one element,
while the algorithm $\mathcal{M}$ is randomized
with its non-deterministic output $\mathcal{M}(S)$
which belongs to $Range(\mathcal{M})$.
In what follows,
we relax this definition for set operations.

For a data set pair $\{S_{A},S_{B}\}$
that consists of two data sets $S_{A}$ and $S_{B}$,
its neighboring data set pair $\{S_{A},S_{B}\}'$
is defined as either
$\{S_{A}',S_{B}\}$ or $\{S_{A},S_{B}'\}$,
where $S_{A}$ differs in one element with $S_{A}'$,
and $S_{B}$ differs in one element with $S_{B}'$.
The randomized algorithm $\mathcal{\ddot{M}}$
is a set operation process
with a nondeterministic output in the range
$Range(\mathcal{\ddot{M}})$.
Based on the above setting,
the \emph{Differentially Private Set Operations}
($\epsilon$-DPSO)
is formally defined as

\begin{definition}[$\epsilon$-DPSO]\label{def:DPSO}
A randomized set operation algorithm $\mathcal{\ddot{M}}$
gives $\epsilon$-differential privacy
if for all neighbouring data set pairs $\{S_{A},S_{B}\}$ and
$\{S_{A},S_{B}\}'$ differing on at most one element,
and all $O \subseteq Range(\mathcal{\ddot{M}})$,
we have
\begin{center}
$Pr[\mathcal{\ddot{M}}(\{S_{A},S_{B}\}) \in O] \leq e^{\epsilon}
\cdot Pr[\mathcal{\ddot{M}}(\{S_{A},S_{B}\}') \in O]$.
\end{center}
\end{definition}

Next,
for a randomized set operation algorithm $\mathcal{\ddot{M}}$,
we observe that although all its possible outputs
belong to $O$,
there maybe exist a narrower outputs set $O_{\sigma}$ that
includes the most possible outputs of the algorithm.
For example,
in the \texttt{MH-JSC},
the probability for its output $J_{mh}(S_{A},S_{B})$ to
be in the range $[\sigma-\theta, \sigma+\theta]$
could be relatively high if given appropriate parameters,
as shown in Eq.~\ref{equ:approx.prob}.
Therefore,
as a condition,
if we only focus on the most possible outputs $O_{\sigma}$
instead of all possible outputs $O$ of a
randomized set operation algorithm $\mathcal{\ddot{M}}$,
the \emph{Conditional} $\epsilon$-DPSO can be further defined as

\begin{definition}[Conditional $\epsilon$-DPSO]\label{def:con-DPSO}
A randomized set operation algorithm $\mathcal{\ddot{M}}$
gives conditional $\epsilon$-differential privacy
if for all neighbouring data set pairs
$\{S_{A},S_{B}\}$ and $\{S_{A},S_{B}\}'$
differing on at most one element,
and for the most possible outputs
$O_{\sigma} \subseteq Range(\mathcal{\ddot{M}})$,
\begin{center}
$Pr[\mathcal{\ddot{M}}(\{S_{A},S_{B}\}) \in O_{\sigma}] \leq e^{\epsilon}
\cdot Pr[\mathcal{\ddot{M}}(\{S_{A},S_{B}\}') \in O_{\sigma}]$.
\end{center}
\end{definition}

\subsection{Privacy Analysis of MH-JSC}

Here,
we will show that the \texttt{MH-JSC} satisfies
the conditional $\epsilon$-DPSO:

\begin{theorem}
The \texttt{MH-JSC} satisfies the conditional $\epsilon$-DPSO.
\end{theorem}

\begin{proof}
Assume that $\sigma$ and $\sigma'$ are the value of
$J(S_{A},S_{B})$ and $J(S_{A},S_{B})'$,
respectively,
and all sets $S_{A}$, $S_{A}'$, $S_{B}$ and $S_{B}'$
have the same size,
that is,
$|S_{A}|=|S_{A}'|=|S_{B}|=|S_{B}'|=n$,
according to equation~\ref{equ:Jsim},
we have
\begin{dgroup*}
\begin{dmath*}
J(S_{A},S_{B})=\frac{|S_{A} \cap S_{B}|}{|S_{A} \cup S_{B}|}=\sigma
\end{dmath*}
\begin{dsuspend}
and
\end{dsuspend}
\begin{dmath*}
J(S_{A}',S_{B})=\frac{|S_{A}' \cap S_{B}|}{|S_{A}' \cup S_{B}|}=\sigma'
\end{dmath*}
\begin{dsuspend}
or
\end{dsuspend}
\begin{dmath*}
J(S_{A},S_{B}')=\frac{|S_{A} \cap S_{B}'|}{|S_{A} \cup S_{B}'|}=\sigma'
\end{dmath*}.
\end{dgroup*}

Then we can have $|\sigma-\sigma'|\leq \frac{1}{n}$,
because the maximum change of the numerator between $\sigma$ and $\sigma'$
is $1$ and the minimum of the denominator between $\sigma$ and $\sigma'$ is $n$.

According to Eq.~\eqref{equ:approx.prob} and above conclusion,
we have

\begin{dmath*}
\frac{Pr[J_{mh}({S_{A},S_{B}}) \in O_{\sigma}]}
{Pr[J_{mh}({S_{A},S_{B}}') \in O_{\sigma}]}
= \frac{Pr[J_{mh}({S_{A},S_{B}}) \in [\sigma-\theta,\sigma+\theta]}
{Pr[J_{mh}({S_{A},S_{B}}') \in [\sigma-\theta,\sigma+\theta]]}
\leq \frac{Pr[J_{mh}({S_{A},S_{B}}) \in [\sigma-\theta,\sigma+\theta]}
{Pr[J_{mh}({S_{A},S_{B}}') \in [\sigma'-(\theta+\frac{1}{n}),\sigma'+(\theta+\frac{1}{n})]}
= \frac{p(K,\sigma,\theta)}{p(K,\sigma',\theta+\frac{1}{n})}
= \frac{\sum_{K(\sigma-\theta) \leq t \leq K(\sigma+\theta)}
\binom{K}{t}(\sigma)^{t}(1-\sigma)^{K-t}}{\sum_{K(\sigma'-\theta-\frac{1}{n})
\leq t' \leq K(\sigma'+\theta+\frac{1}{n})}
\binom{K}{t'}(\sigma')^{t'}(1-\sigma')^{K-t'}}
= e^{\epsilon}.
\end{dmath*}

Based on above formula derivations,
we can conclude that
if we use the $O_{\sigma}$ instead of $O$
to represent the most possible outputs of $\mathcal{\ddot{M}}$,
the computation process of
\emph{MinHash-based Jaccard Similarity} satisfies
the conditional $\epsilon$-DPSO with
$\epsilon = ln(\frac{p(K,\sigma,\theta)}{p(K,\sigma',\theta+\frac{1}{n})})$.
\end{proof}

Since the above privacy property is based on
the observation of a particular subset of
the output space of \texttt{MH-JSC},
the \texttt{MH-JSC} still cannot achieve
the strict differential privacy in which
the privacy property should be maintained across all the output space.
That is to say,
the randomness within the \texttt{MH-JSC}
can only lead to a limited indistinguishability of its outputs,
and external noise is still required for \texttt{MH-JSC}
to achieve the $\epsilon$-differential privacy,
as shown in Section~\ref{sec-private-jaccard-similarity-computation}.

\section{Private Jaccard Similarity Computation}
\label{sec-private-jaccard-similarity-computation}

In this section,
we propose a \emph{\textbf{Priv}ate \textbf{Min}Hash-based
Jaccard Similarity Computation} (\texttt{PrivMin}) algorithm to
achieve the strict differential privacy in \texttt{MH-JSC}.

\subsection{Algorithm Overview}

The $\texttt{PrivMin}$ algorithm aims to
release the MinHash-based Jaccard similarity
between any two cloud users
by the \emph{Profile Perturbation} approach,
which ensures that each user's private information can be
protected from the passive attack similar to the one in Example~\ref{exp:attack}.
That is,
based on observation of the released similarity,
potential adversaries cannot re-identify
the elements in the original user profiles.
The rationale for $\texttt{PrivMin}$ algorithm
is shown in Fig.~\ref{fig:algorithm-rationale}.

\begin{figure}[htbp]
\centering
\includegraphics[scale=0.5]{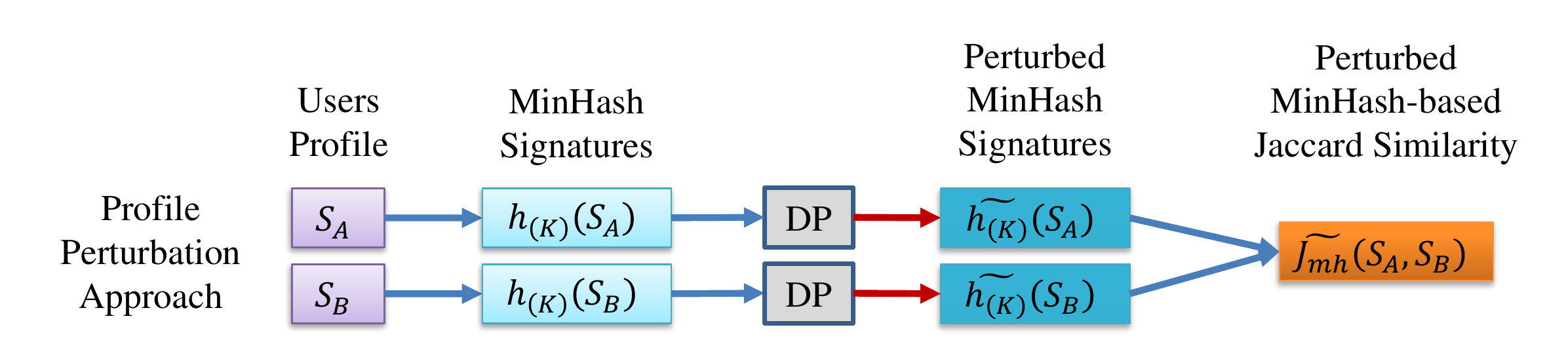}
\caption{Rationale for \texttt{PrivMin} Algorithm}
\label{fig:algorithm-rationale}
\end{figure}

We examine the minimum hash value computation process
within the MinHashing phase,
and add the \emph{Exponential} noise to
the original MinHash signatures
through leveraging the \emph{Randomized Response} strategy.
From the perturbed MinHash signatures,
the adversary cannot infer the sensitive information
within the users' input profiles.
Specifically,
we conceptualize the $\texttt{PrivMin}$ algorithm
into two private operations:
\begin{description}
\item[Private MinHash Value Generation]
Based on the \emph{Exponential} mechanism,
this operation privately selects the minimum hash value
in each step within the MinHashing phase.
By default this operation will be executed in all the $K$ steps
and then the perturbed MinHash signature will be generated.

\item[Randomized MinHashing Steps Selection]
In the generation of perturbed MinHash signature,
this operation privately shrinks the number $K$ into $m$
by the \emph{Randomized Response} technique,
so that the total added \emph{Exponential} noise
is tightly controlled.

\end{description}

Details for the \emph{Private MinHash Value Generation}
is presented in Section~\ref{sec-private-minhash-signature-generation},
followed by the \emph{Randomized MinHashing Steps Selection}
in Section~\ref{sec-randomized-minhashing-steps-selection}.

\subsection{Private MinHash Value Generation}
\label{sec-private-minhash-signature-generation}

In this operation,
we attempt to add the \emph{Exponential} noise
through the minimum hash value computation process
within the MinHashing phase,
and then generate the perturbed MinHash signatures
$\widetilde{h_{(K)}}(S_{A})$ and $\widetilde{h_{(K)}}(S_{B})$,
for the profiles $S_{A}$ and $S_{B}$.
The intuition behind this operation is that
we aim to add just enough noise
by leveraging the internal noise in \texttt{MH-JSC}.
In this way,
by using the perturbed MinHash signatures,
the final similarity would also be a noisy version
from which the semi-honest users
cannot successfully launch a passive attack.

More specifically,
as proved in Section~\ref{sec-dpso},
the \texttt{MH-JSC} satisfies the \emph{Conditional $\epsilon$-DPSO}
because of its internal randomness.
Herein,
we first show that
the MinHashing phase produces such randomness and it
only satisfies the differential privacy in certain situations.

\begin{lemma}
\label{lem:privacy_minhash}
The MinHashing phase only satisfies
the $\epsilon$-differential privacy
at certain situations in which
the element difference between $S$ and $S'$
has an impact on the equality of their minimum hash value.
\end{lemma}

\begin{proof}
Following the steps described in Section~\ref{sec-minhash},
in the MinHashing phase with $K$ hash functions,
when given neighbour profiles $S$ and $S'$
which differ only in one element,
their hash value sets for each hash function $k$ are
generated and denoted as $h_{k}(S)$ and $h_{k}(S')$.
Next,
the minimum hash values $min\{h_{k}(S)\}$
and $min\{h_{k}(S')\}$ are selected for
further construction of the MinHash signatures
$h_{(K)}(S)$ and $h_{(K)}(S')$.

If the value of element that differentiates
$S$ from $S'$ has no impact on the equality
of their minimum hash value under a hash function $k$,
that is,
$min\{h_{k}(S)\} = min\{h_{k}(S')\}$,
we have
\begin{equation*}
Pr[min\{h_{k}(S)\} \in O] =
Pr[min\{h_{k}(S')\} \in O]
for O \subseteq Range(min\{\cdot\}),
\end{equation*}
which satisfies the $\epsilon$-differential privacy
where $\epsilon = 0$,
as shown in Definition~\ref{def:DP}.

However,
if the value of element that differentiates
$S$ from $S'$ does have an impact on the equality
of their minimum hash value under a hash function $k$,
that is,
$min\{h_{k}(S)\} \neq min\{h_{k}(S')\}$,
the MinHashing phase cannot guarantee
the $\epsilon$-differential privacy in any degree.
\end{proof}

Based on the above result,
we design the \emph{Private MinHash Value Generation} algorithm
by adopting the \emph{Exponential} mechanism
to privately select the minimum hash value
in the steps within MinHashing.
Specifically,
when adopting the \emph{Exponential} mechanism,
the \emph{Score Function} $q$ should be carefully defined.
Following the suggestions in~\cite{JorgensenYC15ICDE},
we use the notation $h_{k}(S) \oplus h_{k}(S')$
to denote the certain set of elements
in which $h_{k}(S)$ and $h_{k}(S')$ are differ.
It is noted that herein
$h_{k}(S)$ and $h_{k}(S')$ are not limited to the neighbouring datasets.
Then,
the \emph{Score Function} can be defined as
\begin{equation}
q(h_{k}(S),\psi)=\max \limits_{min\{h_{k}(S')\}
=\psi} -| h_{k}(S) \oplus h_{k}(S')|,
\end{equation}
and the \emph{sensitivity} of $q$ is $\bigtriangleup q = 1$.

The following example provides concrete cases
showing how to compute the score of the candidates in a given $h_{k}(S)$.
\begin{description}
\item[Example 2]
Assume $h_{k}(S)=\{11,13,15,16,19\}$,
then we have
$q(h_{k}(S),10)=q(h_{k}(S),12)=q(h_{k}(S),13)=-1$
because changing only one element in $h_{k}(S)$
is enough to make these minimum values become the true
answers of $min\{h_{k}(S')\}$.
In the meanwhile,
we have $q(h_{k}(S),16)=-3$ because making $16$ become
the true answer of $min\{h_{k}(S')\}$ would require
changing three elements in $h_{k}(S)$.
\end{description}

\begin{algorithm}[htbp]
\small
\caption{Private MinHash Value Generation}
\label{alg:private_minhash}
\begin{algorithmic}[1]
\REQUIRE
        Profile $S=\{ s_1, s_2, ... , s_N \}$,
        $K$ hash functions,
        overall privacy budget $\epsilon$.
\ENSURE
        Perturbed MinHash signature $\widetilde{h_{(K)}}(S)$.
\STATE
Initialize a null vector $\widetilde{h_{(K)}}(S)$;
\FOR {$k \leftarrow 1...K$}
\FOR {$n \leftarrow 1...N$}
\STATE  Compute the hash value $h_{k}(s_n)$;
\ENDFOR
\STATE  Construct the hash values set
$h_{k}(S)=\{ h_{k}(s_1), h_{k}(s_2), ..., h_{k}(s_N) \}$;
\STATE  Select the minimum hash value $min\{h_{k}(S)\}$
with probability proportional to $\textsf{exp}(\frac{\epsilon
q(h_{k}(S),\psi)}{2K\bigtriangleup q})$;
\STATE  Append to $\widetilde{h_{(K)}}(S)$;
\ENDFOR
\RETURN $\widetilde{h_{(K)}}(S)$
\end{algorithmic}
\end{algorithm}

The details of \emph{Private MinHash Value Generation}
are shown in Algorithm~\ref{alg:private_minhash}.
It starts with a user profile $S$,
the overall privacy budget $\epsilon$ and $K$ hash functions.
First,
the cloud platform initializes a null vector
$\widetilde{h_{(K)}}(S)$ (Line 1).
Then,
for the $k^{th} \ (1 \leq k \leq K)$ hash function $h_{k}$,
the cloud platform computes the related hash value set
$h_{k}(S)$ (line 2-4).
Next,
to achieve the differential privacy,
the cloud platform selects the minimum hash value
$min\{h_{k}(s_n)\}$
with probability proportional to
$\textsf{exp}(\frac{\epsilon
q(h_{k}(S),\psi)}{2K\bigtriangleup q})$,
and appends the selected value
to \emph{Perturbed MinHash Signature}
$\widetilde{h_{(K)}}(S)$ (Line 5-8).
The generated \emph{Perturbed MinHash Signature}
can be further used to compute
the MinHash-based Jaccard similarity.

\subsection{Randomized MinHashing Steps Selection}
\label{sec-randomized-minhashing-steps-selection}

Although the above operation relies on
the minimum hash value computation process
within the MinHashing phase
and makes full use of the internal noise,
its utility remains far from acceptable.
This is because when all the steps within the MinHashing phase
adopt the \emph{Exponential} mechanism,
the cumulative noise would seriously distort the accuracy of outputs.
Fortunately,
we find that
if we use the \emph{Randomized Response} technique
to select steps for adopting the \emph{Exponential} mechanism,
the combined algorithm will successfully achieve
both in rigorous differential privacy and in acceptable utility.
More specifically,
the \emph{Randomized MinHashing Steps Selection} operation
consists of two main steps which are described as the following.
\begin{enumerate}
\item
The operation randomly selects several steps out of
all the $K$ steps within the MinHashing phase with probability $P_{r}$.
For the convenience of recording this result,
we maintain an \emph{Original Flip Vector}
$\overrightarrow{V}$
in which the binary value in the $k^{th}$ ($1 \leq k \leq K$) place represents
whether the $k^{th}$ step is initially implemented
with the \emph{Exponential} mechanism.
As the above generated vector should be prevented from
potential attack and should not be directly used
in the \emph{Private MinHash Value Generation} operation,
a \emph{Perturbed Flip Vector} $\overrightarrow{V'}$
will be generated in the next step.

\item
Aiming to satisfy the differential privacy in this step,
we generate the \emph{Perturbed Flip Vector}
$\overrightarrow{V'}$
with \emph{Randomized Response} technique
as described in Section~\ref{sec-rr} and
similar to the \emph{Permanent Randomized Response} proposed in~\cite{ErlingssonPK14CCS}.
The \emph{Perturbed Flip Vector}
$\overrightarrow{V'}$ will indicate
which steps within the MinHashing phase will finally be
implemented with the \emph{Exponential} mechanism.
Specifically,
given an \emph{Original Flip Vector} $\overrightarrow{V}$,
for the value $V_{k}$ in each bit $k \in [0,K]$ of $\overrightarrow{V}$,
this step generates a perturbed binary value $V_{k}'$
which equals to:
\begin{equation*}
V_{k}' =
\begin{cases}
1, & \text{with probability $1/2 P_{t}$},\\
0, & \text{with probability $1/2 P_{t}$},\\
V_{k}, & \text{with probability $1 - P_{t}$},
\end{cases}
\end{equation*}
where $P_{t}$ is the threshold probability to flip the original binary value.
In this way,
the \emph{Perturbed Flip Vector} $\overrightarrow{V'}$
will be generated.
It is noted that we directly set the $P_{r} = P_{t}$
in our experiments as a default setting.
An intuition description of this procedure is shown in
Fig~\ref{fig:perturbed-flip-vector}.
\end{enumerate}

\begin{figure}[htbp]
\centering
\includegraphics[scale=0.5]{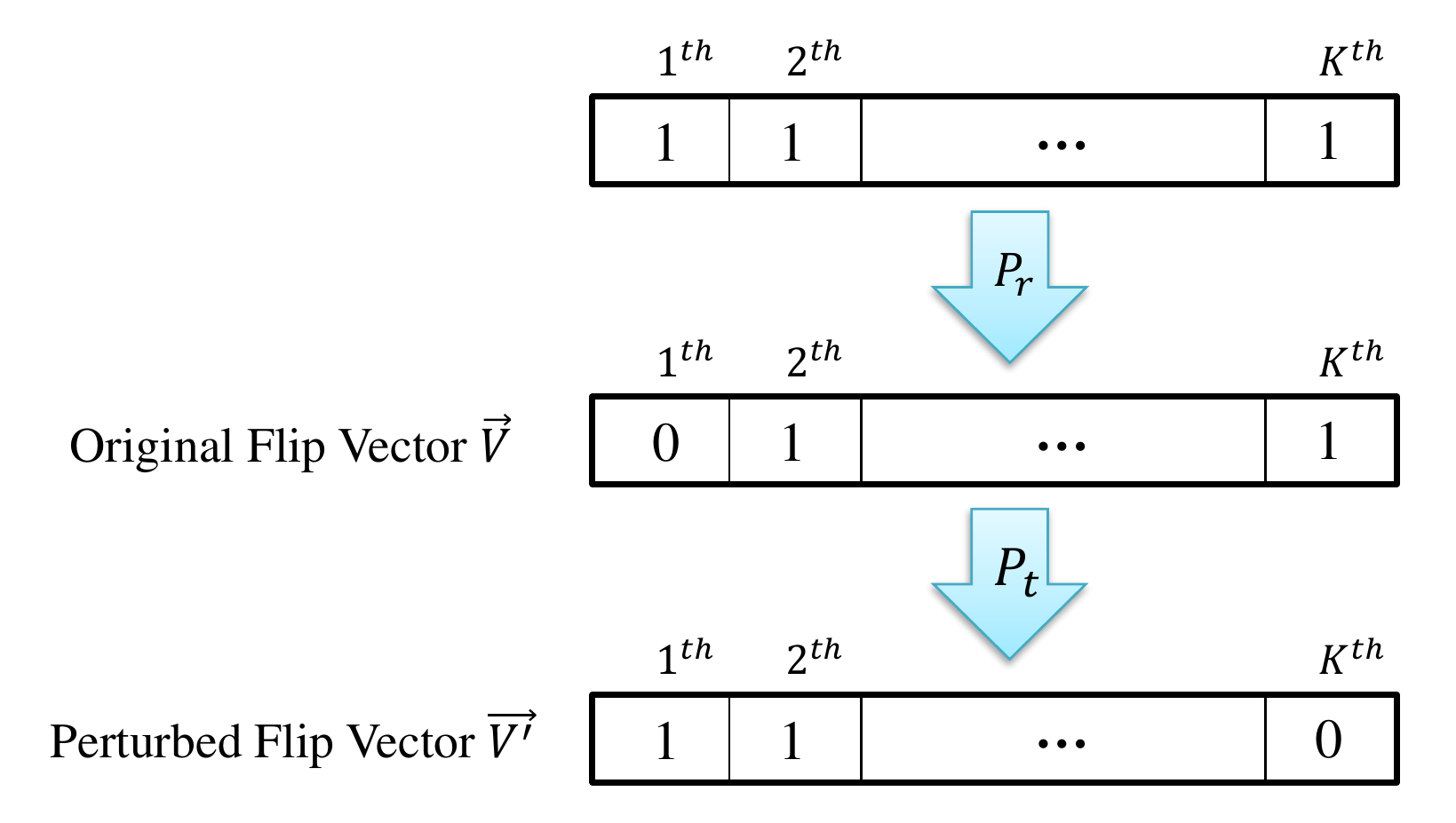}
\caption{Generation of Perturbed Flip Vector}
\label{fig:perturbed-flip-vector}
\end{figure}

By incorporating the \emph{Private MinHash Value Generation}
with the \emph{Randomized MinHashing Steps Selection},
the pseudocode of the \texttt{PrivMin} algorithm
is given in Algorithm~\ref{alg:privmin}.
It starts with a profile $S$ of a user,
the overall privacy budget $\epsilon$ and $K$ hash functions.
\textbf{Firstly,}
the cloud platform initializes a null vector
$\widetilde{h_{(K)}}(S)$ (Line 1) and
divides the overall privacy budget $\epsilon$
into two equal parts,
$\epsilon_{1}$ and $\epsilon_{2}$ (Line 2-3).
The former one is used to calculate the value of the
probability threshold $P_{t}$
used in \emph{Randomized MinHashing Steps Selection}
while the latter will be assigned
to the \emph{Private MinHashing Signature Generation}.
\textbf{Secondly,}
the \emph{Perturbed Flip Vector} generation
would be triggered as described above (Line 4-5).
\textbf{Thirdly,}
according to the generated $\overrightarrow{V'}$,
the \emph{Exponential} mechanism would be deployed on
the marked steps within the MinHashing phase (Line 6-15).
And for the steps which are not marked
in $\overrightarrow{V'}$,
compute their original outputs (Line 16-22).
\textbf{Finally,}
the computed minimum hash values of each step
would be appended to \emph{Perturbed MinHash Signature}
$\widetilde{h_{(K)}}(S)$.
The generated \emph{Perturbed MinHash Signature}
can be further used to compute
the MinHash-based Jaccard similarity.

\begin{algorithm}[htbp]
\small
\caption{\texttt{PrivMin} algorithm}
\label{alg:privmin}
\begin{algorithmic}[1]
\REQUIRE
        Profile $S=\{ s_1, s_2, ... , s_N \}$,
        $K$ hash functions,
        overall privacy budget $\epsilon$.
\ENSURE
        Perturbed MinHash signature $\widetilde{h_{(K)}}(S)$.
\STATE  Initialize a row vector $\widetilde{h_{(K)}}(S)$;
\STATE  $\epsilon_{1} \leftarrow \epsilon / 2$,
        $P_{t} \leftarrow \frac{e^{\epsilon_{1} / K}}{1 + e^{\epsilon_{1} / K}}$;
\STATE  $\epsilon_{2} \leftarrow \epsilon / 2$;
\STATE  Construct the \emph{Original Flip Vector} $\overrightarrow{V}$ by randomly choosing
        its $K$ elements
        $r_{1},r_{2},...,r_{K}$ from $\{0,1\}^{1}$ with probability $P_{r}$;
\STATE  Construct the \emph{Perturbed Flip Vector} $\overrightarrow{V'}$ by implementing
        \emph{Randomized Response} technique described in Section~\ref{sec-randomized-minhashing-steps-selection};
\STATE  m $\leftarrow$ Compute the numbers of the elements in $\overrightarrow{V'}$
        that is equal to $0$;
\STATE  $\epsilon' \leftarrow \epsilon_{2} / m$;
\FOR {$k \leftarrow 1...K$}
\IF{$V_{k}' = 0$}
\FOR {$n \leftarrow 1...N$}
\STATE  Compute the hash value $h_{k}(s_n)$;
\ENDFOR
\STATE  Construct the hash values set
        $h_{k}(S)=\{ h_{k}(s_1), h_{k}(s_2), ..., h_{k}(s_N) \}$;
\STATE  Select the minimum hash value $min\{h_{k}(S)\}$
        with probability proportional to $\textsf{exp}(\frac{\epsilon'
        q(h_{k}(S),\psi)}{2 \bigtriangleup q})$;
\STATE  Append to $\widetilde{h_{(K)}}(S)$;
\ELSE
\FOR {$n \leftarrow 1...N$}
\STATE  Compute the hash value $h_{k}(s_n)$;
\ENDFOR
\STATE  Construct the hash values set
        $h_{k}(S)=\{ h_{k}(s_1), h_{k}(s_2), ..., h_{k}(s_N) \}$;
\STATE  Select the minimum value $min\{h_{k}(S)\}$;
\STATE  Append to $\widetilde{h_{(K)}}(S)$;
\ENDIF
\ENDFOR
\RETURN $\widetilde{h_{(K)}}(S)$
\end{algorithmic}
\end{algorithm}

\section{Algorithm Analysis}\label{sec-analysis}

The proposed \texttt{PrivMin} algorithm aims to
achieve the differential privacy while
maintaining an acceptable utility.
In this section,
we will prove that
the algorithm satisfies $\epsilon$-differential privacy
and then provide the utility analysis.

\subsection{Privacy Analysis}\label{sec-privacy-analysis}

Based on the \emph{Sequential Composition} and
the \emph{Parallel Composition} as in
Theorem~\ref{thr:sequential-composition} and
~\ref{thr:parallel-composition},
we have the following theorem on the privacy guarantee of the proposed algorithm.

\begin{theorem}\label{thr:output-perturbation}
The \texttt{PrivMin} algorithm satisfies
$\epsilon$-differential privacy.
\end{theorem}

\begin{proof}
Two independent private operations
of the \texttt{PrivMin} algorithm can respectively
satisfy relevant level of differential privacy as follows:
\begin{enumerate}
\item[-]
Based on the proofs in~\cite{ErlingssonPK14CCS},
since we adopt a similar randomized response approach
as the \emph{Permanent Randomized Response}
in~\cite{ErlingssonPK14CCS},
we can conclude that
the \emph{Randomized MinHashing Steps Selection} operation
satisfies $\epsilon_{1}$-differential privacy
where $\epsilon_{1} = Kln(\frac{P_{t}}{1-P_{t}})$.

\item[-]
As the \emph{Private MinHash Value Generation} operation
adopts the \emph{Exponential} mechanism successively in
the privately selected steps within the MinHashing phase,
this operation satisfies $\epsilon_{2}$-differential privacy
since the $m$ selected steps respectively achieve
$\frac{\epsilon_{2}}{m}$-differential privacy.
\end{enumerate}
Consequently,
according to the \emph{Sequential Composition},
we can conclude that
the \texttt{PrivMin} algorithm
satisfies $\epsilon$-differential privacy
where $\epsilon =\epsilon_{1}+\epsilon_{2}$.
\end{proof}

\subsection{Utility Analysis}\label{sec-utility analysis}

Here we adopt \emph{$(\alpha,\delta)-usefulness$}
to measure the \emph{Semantic Loss} in
each step of the \emph{Private MinHash Value Generation} operation.

\begin{theorem}
For all $\delta > 0$,
with probability at least $1- \delta$,
the \emph{SLoss} of the MinHash signatures
in the \emph{Private MinHash Value Generation}
operation is less than $\alpha$.
When
\begin{equation*}
1 - \frac{3}{2}P_{t} + P_{t}^{2} \leq \delta \alpha,
\end{equation*}
where $P_{t} = \frac{e^{\epsilon_{1} / K}}{1 + e^{\epsilon_{1} / K}}$,
and the \emph{Private MinHash Value Generation} operation
is satisfied with
$(\alpha,\delta)$-useful.
\end{theorem}

\begin{proof}
According to Markov's inequality,
we have
\begin{equation}\label{equ:marlkov-inequality}
Pr(SLoss > \alpha) \leq \frac{E(SLoss)}{\alpha}
\end{equation}
For each minimum hash value $min\{h_{k}(S)\}$
in $\widetilde{h_{(K)}}(S)$,
the probability of ``unchange"
in the randomized private selection
is proportional to
\begin{dmath*}
P_{r} \cdot (\frac{1}{2}P_{t} + 1 - P_{t})
+ (1 - P_{r}) \cdot \frac{1}{2}P_{t}
= P_{r} - P_{r}P_{t} + \frac{1}{2}P_{t}
= \frac{3}{2}P_{t} - P_{t}^{2}.
\end{dmath*}
Therefore,
we have
\begin{equation*}
E(SLoss) = \sum_{min\{h_{k}(S)\} \in h_{(K)}(S)}
\frac{d(min\{h_{k}(S)\},\widehat{min}\{h_{k}(S)\})}{max~d \cdot | h_{(K)}(S) |}
(1 - \frac{3}{2}P_{t} + P_{t}^{2}).
\end{equation*}
According to Eq.~\eqref{equ:marlkov-inequality},
the evaluation of the \emph{SLoss} is
\begin{equation*}
Pr(SLoss > \alpha) \leq \frac
{\sum_{min\{h_{k}(S)\} \in h_{(K)}(S)} d(min\{h_{k}(S)\},\widehat{min}\{h_{k}(S)\})
(1 - \frac{3}{2}P_{t} + P_{t}^{2})}
{max~d \cdot | h_{(K)}(S) | \cdot \alpha}.
\end{equation*}
When we take the maximal $d(min\{h_{k}(S)\},\widehat{min}\{h_{k}(S)\}) = K$,
it can be simplified as
\begin{equation}
Pr(SLoss > \alpha) \geq 1 - \frac{1 - \frac{3}{2}P_{t} + P_{t}^{2}}{\alpha}.
\end{equation}
Let
\begin{equation*}
1 - \frac{1 - \frac{3}{2}P_{t} + P_{t}^{2}}{\alpha} \geq 1 - \delta,
\end{equation*}
thus
\begin{equation}
1 - \frac{3}{2}P_{t} + P_{t}^{2} \leq \delta \alpha,
\end{equation}
where $P_{t} = \frac{e^{\epsilon_{1} / K}}{1 + e^{\epsilon_{1} / K}}$.
\end{proof}
The proof shows that the Semantic Loss of
the \emph{Private MinHash Value Generation} operation
mainly depends on the privacy budget $\epsilon_{1}$
and the hash function number $K$.

\section{Experiment and Analysis}\label{sec-experiment}

In this section,
we conduct experiments to examine the performance
of the proposed \texttt{PrivMin} algorithm
by answering the following questions:
\begin{description}
\item[How does the \texttt{PrivMin} algorithm preserve the utility?]
The \texttt{PrivMin} algorithm aims to release
Jaccard similarities with acceptable utility.
In Section~\ref{sec-experimental-performance},
we will investigate its performance in terms of
\emph{$F1$ Score} and \emph{Mean Squared Error} (MSE)
on the released similarities,
and compare it with the \emph{Baselin}e algorithm
and \texttt{MH-JSC}.

\item[How will the main parameters impact on the performance of it?]
The \texttt{PrivMin} algorithm has two parameters
$\epsilon$ and $K$:
$\epsilon$ controls the privacy level of algorithms;
and $K$ determines the total number of Hash functions
which are used in MinHashing phase.
In Section~\ref{sec-impact-of-privacy-budget}
and~\ref{sec-impact-of-hash-fuction-number},
we will investigate and analyze their impacts on
the involved three algorithms.
\end{description}

\subsection{Experiment Setting}

\subsubsection{Datasets and Configuration}

We evaluate the compared algorithms on four real textual datasets:
\begin{itemize}
\item
\textsf{Alpine Dale}:
The \textsf{Alpine Dale} dataset
\footnote{http://www.inf.ed.ac.uk/teaching/courses/tts/assessed/assessment3.html}
was retrieved from the course website
of ``text technologies for data science''
by the University of Edinburgh,
and includes $10000$ news stories for plagiarism detection.
In the following experiments,
we will use a subset with $1000$ records.

\item
\textsf{BBC Sport}:
This dataset was derived from \emph{Insight Project Resources}
~\footnote{http://mlg.ucd.ie/datasets/bbc.html}.
It contains $737$ documents from the BBC Sport website
corresponding to sports news articles in five topical areas
from $2004-2005$.

\item
\textsf{Opinosis}~\cite{White:2015:WSE:2838931.2838932,ganesan2010opinosis}:
This dataset from \emph{Paraphrase Grouped Corpora}
\footnote{http://white.ucc.asn.au/resources/paraphrase_grouped_corpora/}
is a subset of the \texttt{Opinosis} corpus
\footnote{http://kavita-ganesan.com/opinosis-opinion-dataset}.
It contains $669$ sentences
which were manually grouped according to their meaning.

\item
\textsf{MSRP}~\cite{White:2015:WSE:2838931.2838932,dolan2004unsupervised}:
This dataset from \emph{Paraphrase Grouped Corpora}
is a subset of the the Microsoft Research Paraphrase corpus
\footnote{http://research.microsoft.com/en-us/downloads/607d14d9-20cd-47e3-85bc-a2f65cd28042/}.
It contains $859$ sentences
which was automatically grouped according to
its original manually annotated meaning.

\end{itemize}

The involved three algorithms are implemented
in Python 2.7 based on the code by Chris McCormick
\footnote{http://mccormickml.com/2015/06/12/minhash-tutorial-with-python-code/}.
All the experiments are conducted on
an Intel Core i5-3210M 2.50GHz PC with 6GB memory.
In each experiment,
every algorithm is executed $10$ times,
and its average score is reported.

\subsubsection{Experiment Parameters}

We consider two parameters $\epsilon$ and $K$
since the performance of algorithms could be affected by them:
\begin{description}
\item[the privacy budget $\epsilon$]
Although the tradeoff between the privacy budget $\epsilon$
and the utility under the naive \emph{Laplace} mechanism
and the \emph{Exponential} mechanism is known,
we also expect to discover the situation
in which the \emph{Randomized Response} cooperates
with the \emph{Exponential} mechanism.

\item[the number of Hash functions $K$]
Although it is clear that
a smaller number of hash functions
may lead to worse accuracy in similarity,
we expect that the \texttt{PrivMin} algorithm
to perform well when $K$ is relatively small.
\end{description}

In our experiments,
we will vary above two parameters to study their impacts
on the involved algorithms,
in terms to the metrics as mentioned in Section~\ref{sec-utility-metrics}.

\subsubsection{Utility Metrics}\label{sec-utility-metrics}

We adopt the \emph{$F1$ score}
and \emph{Mean Squared Error} (MSE)
to measure the utility performance
among the proposed \texttt{PrivMin} algorithm,
the \emph{Baseline} algorithm and the \texttt{MH-JSC}.
\begin{description}
  \item[$F1$ Score]
The \emph{$F1$ Score} is the harmonic mean of
\emph{Precision $P$} and \emph{Recall $R$},
which can measure the algorithm outputs's accuracy
compared with the given ground truth.
A higher \emph{$F1$ Score} means a better accuracy.
Herein,
the accurate Jaccard similarity
of given two profiles is set as the ground truth.
We aim to investigate the statistical differences
between the released perturbed Jaccard similarity
and the accurate one in several tests.
The \emph{$F1$ Score} can be calculated as
\begin{equation}\label{equ:F1-score1}
F1 = \frac{2 \times P \times R}{P + R},
\end{equation}

\begin{equation}\label{equ:F1-score2}
P = \frac{TP}{TP + FP},
\end{equation}

\begin{equation}\label{equ:F1-score3}
R = \frac{TP}{TP + FN},
\end{equation}
where $TP$ is \emph{true positive},
$FP$ is \emph{false positive},
$TN$ is \emph{true negative},
$FN$ is \emph{false negative}.
Table~\ref{tab:tp-fp-tn-fn} shows the details of setting for these four variables.
According to the specific characteristics of
the textual records within four datasets,
we empirically set the related thresholds as
$0.5$,
$0.4$,
$0.5$ and $0.3$.

\begin{table}[h]  \centering
\small
  \caption{Settings of TP, FP, TN and FN}\label{tab:tp-fp-tn-fn}
  \begin{tabular}{l|l}
\hline
    $TP$
    & \tabincell{c}{the number of test in which \textbf{both} the perturbed Jaccard similarity\\
    and the accurate Jaccard similarity are \textbf{above} a given threshold.}  \\ \hline

    $FP$
    & \tabincell{c}{the number of test in which \textbf{only} the perturbed Jaccard similarity\\
    is \textbf{above} a given threshold and the accurate Jaccard similarity is not.}  \\ \hline

    $TN$
    & \tabincell{c}{the number of test in which \textbf{both} the perturbed Jaccard similarity\\
    and the accurate Jaccard similarity are \textbf{below} a given threshold.}  \\ \hline

    $FN$
    & \tabincell{c}{the number of test in which \textbf{only} the perturbed Jaccard similarity\\
    is \textbf{below} a given threshold and the accurate Jaccard similarity is not.}  \\ \hline
\hline
\end{tabular}
\end{table}

\item[Mean Squared Error (MSE)]
The \emph{Mean Squared Error} (MSE)
is a measure of the quality of an estimator
by calculating the error between the estimator's predicted value and
its accurate value.
A lower MSE means a better accuracy.
The MSE in the following experiments can be calculated as
\begin{equation}\label{equ:mse}
MSE = \frac{1}{n}\sum_{t=1}^{n}(predicted_{t}-accurate_{t})^{2},
\end{equation}
\end{description}
where the $predicted_{t}$ and $accurate_{t}$ are
corresponding to the perturbed Jaccard similarity
and accurate Jaccard similarity,
respectively.
We aim to investigate the numerical errors
between the released perturbed Jaccard similarity
and the accurate one in several tests.

\subsubsection{Compared Algorithms}

We consider a \emph{Baseline} algorithm
and the \texttt{MH-JSC} as the competitors
of the \texttt{PrivMin} algorithm.
\begin{description}
\item[Baseline]
The \emph{Baseline} algorithm is based on
the \emph{Output Perturbation} approach which
introduces differential privacy
by directly adding \emph{Laplacian} noise to
the output similarity $J_{mh}(S_{A},S_{B})$ of \texttt{MH-JSC}:
\begin{equation*}
\widetilde{J_{mh}}(S_{A},S_{B})
=J_{mh}(S_{A},S_{B})
+Laplace(\frac{\Delta J_{mh}}{\epsilon}).
\end{equation*}
The added \emph{Lapacian} noise is calibrated to
the sensitivity of \texttt{MH-JSC} as the following:
\begin{dmath*}
\Delta J_{mh} = \max_{S_{B}, S_{B'} neighbours} \| J_{mh}(S_{A},S_{B})-J_{mh}(S_{A},S_{B'}) \| \\
= \max_{S_{B}, S_{B'} neighbours}
\| \frac{|h_{(K)}(S_{A}) \cap h_{(K)}(S_{B})|-
|h_{(K)}(S_{A}) \cap h_{(K)}(S_{B'})|}{K} \| \\
\leq \max \limits_{S_{A},S_{B}}
\| \frac{|h_{(K)}(S_{A}) \cap h_{(K)}(S_{B})|-
(|h_{(K)}(S_{A}) \cap h_{(K)}(S_{B})| \pm 1)}{K} \| \\
= \frac{1}{K}.
\end{dmath*}
Finally,
the \emph{Baseline} algorithm will release the perturbed similarity
$\widetilde{J_{mh}}(S_{A},S_{B})$.
to the users $U_{A}$ and $U_{B}$.
Since the \emph{Baseline} algorithm intuitively adds coarse-grained noise
to achieve differential privacy,
we expect that it will underperform
the \texttt{PrivMin} algorithm in most cases.

\item[MH-JSC]
The \emph{MinHash-based Jaccard Similarity Computation}
(\texttt{MH-JSC}) can be regarded as a comparative algorithm
which maintains an empirical utility upper bound.
Since \texttt{MH-JSC} does not add any external noise,
we expected that it will outperform both the \texttt{PrivMin}
algorithm and the \emph{Baseline} algorithm in most cases,
and it's performance will also be much closer to that of
the \texttt{PrivMin} algorithm.
\end{description}

\subsection{The Performance of PrivMin}
\label{sec-experimental-performance}

\subsubsection{Impact of Privacy Budget}
\label{sec-impact-of-privacy-budget}

Firstly,
we fix $K=5,10,15,20,25$ and report the utility measures
of different algorithms when varying $\epsilon$
from $0.1$ to $1.0$.
Fig.~\ref{fig:varying-epsilon} shows
the $F1$ score over four datasets
with the change of $\epsilon$.
We observe that the \texttt{PrivMin} algorithm
has higher \emph{$F1$ Scores}
than the \emph{Baseline} algorithm
when given smaller $K$ and $\epsilon$
on all datasets.
Specifically in Fig.~\ref{fig:varying-epsilon}D,
when $K = 5$ and $\epsilon = 0.2$,
\texttt{PrivMin} achieves a \emph{$F1$ Score}
of $0.3060$ while \emph{Baseline} achieves only $0.0006$,
with an improvements by $50900\%$.
When $\epsilon = 0.5$,
\texttt{PrivMin} achieves a \emph{$F1$ Score}
of $0.6881$ and outperforms the \emph{Baseline} by $49050\%$.
The improvements by \texttt{PrivMin} can also be
observed in other subfigures in Fig.~\ref{fig:varying-epsilon}.
For the \emph{Baseline} algorithm,
the larger $\epsilon$,
the higher $F1$ scores.
However,
we observe that the \texttt{PrivMin} algorithm
is not clearly affected by changing the privacy budget.

For the \emph{Mean Squared Error} (MSE),
Table~\ref{tab:mse-comparison} shows that
the \texttt{PrivMin} algorithm generally outperforms the
\emph{Baseline} algorithm with the changing privacy budget.
And in some conditions,
it also maintains a better utility compared with
the \texttt{MH-JSC}.

\begin{figure*}[htbp]
\centering
\includegraphics[scale=0.35]{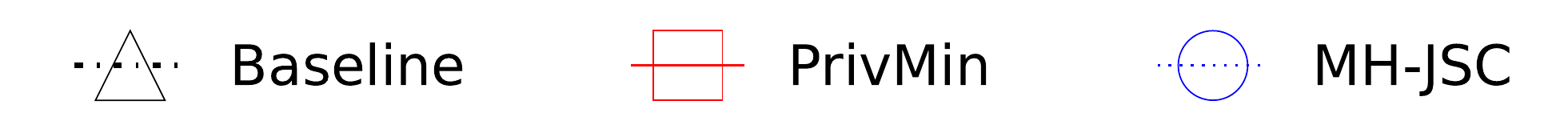}\\
\subfloat[Alpine Dale]{\includegraphics[width = 0.3\linewidth]{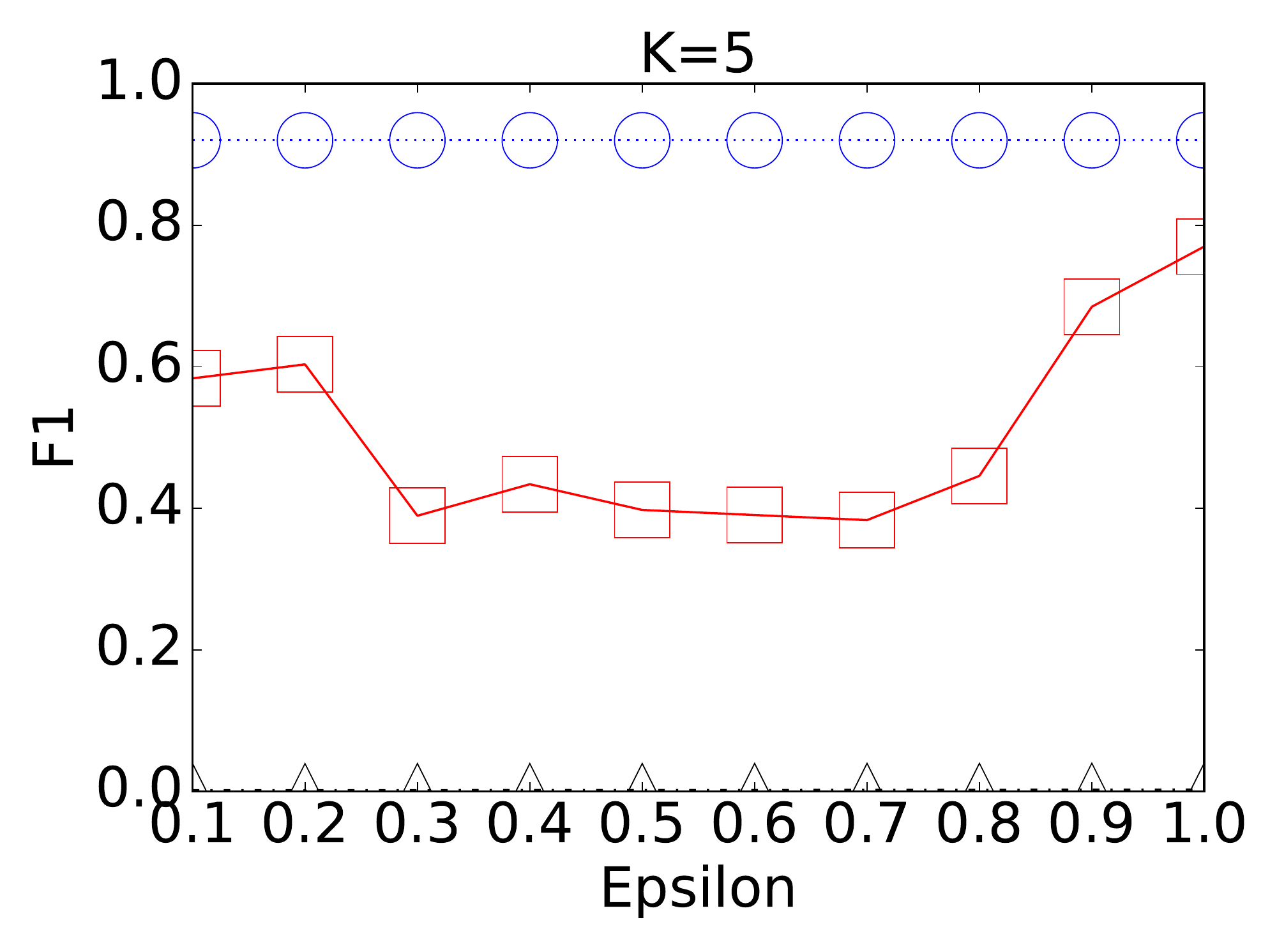}}\
\subfloat[Alpine Dale]{\includegraphics[width = 0.3\linewidth]{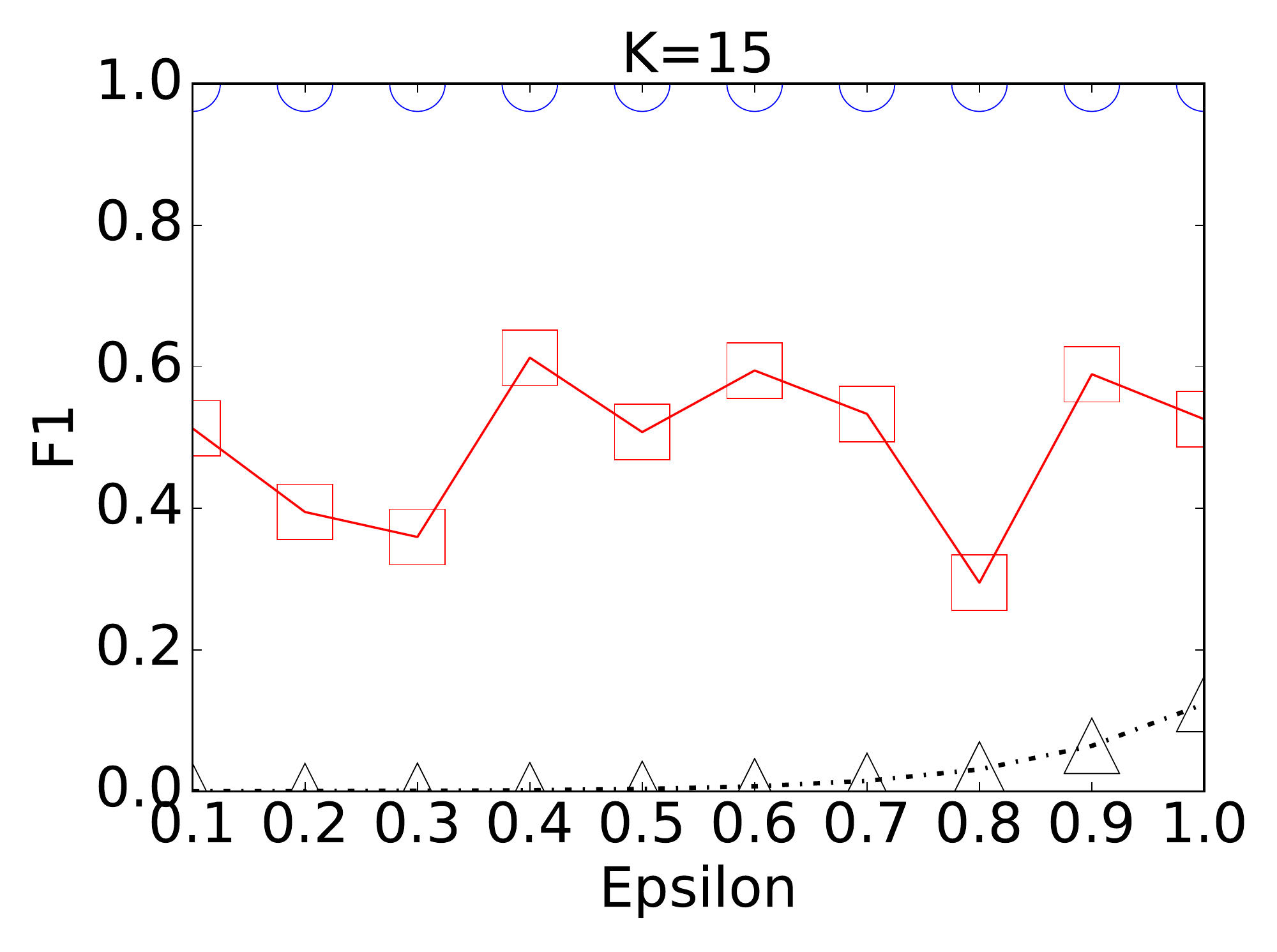}}\
\subfloat[Alpine Dale]{\includegraphics[width = 0.3\linewidth]{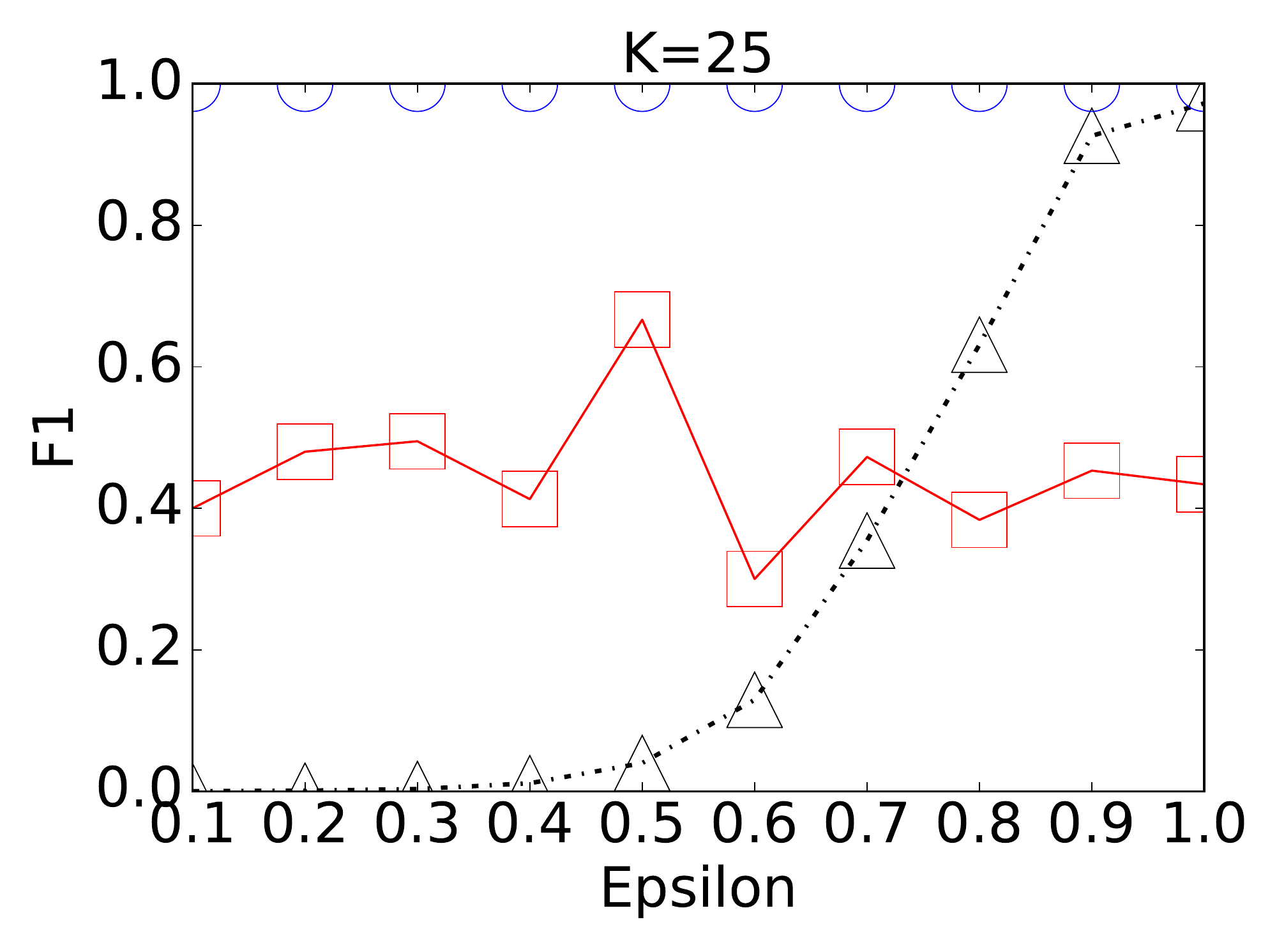}}\\
\subfloat[BBC Sport]{\includegraphics[width = 0.3\linewidth]{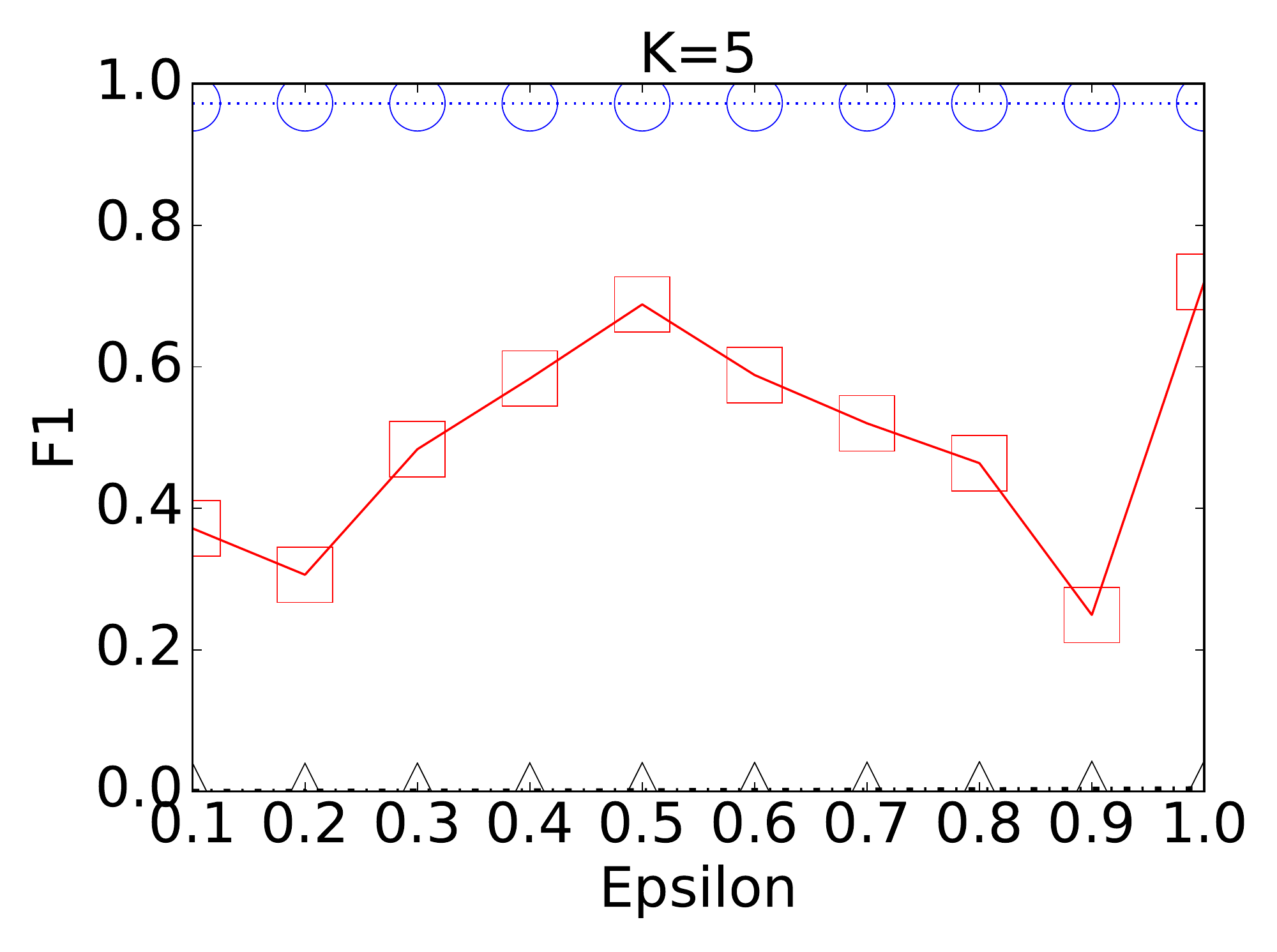}}\
\subfloat[BBC Sport]{\includegraphics[width = 0.3\linewidth]{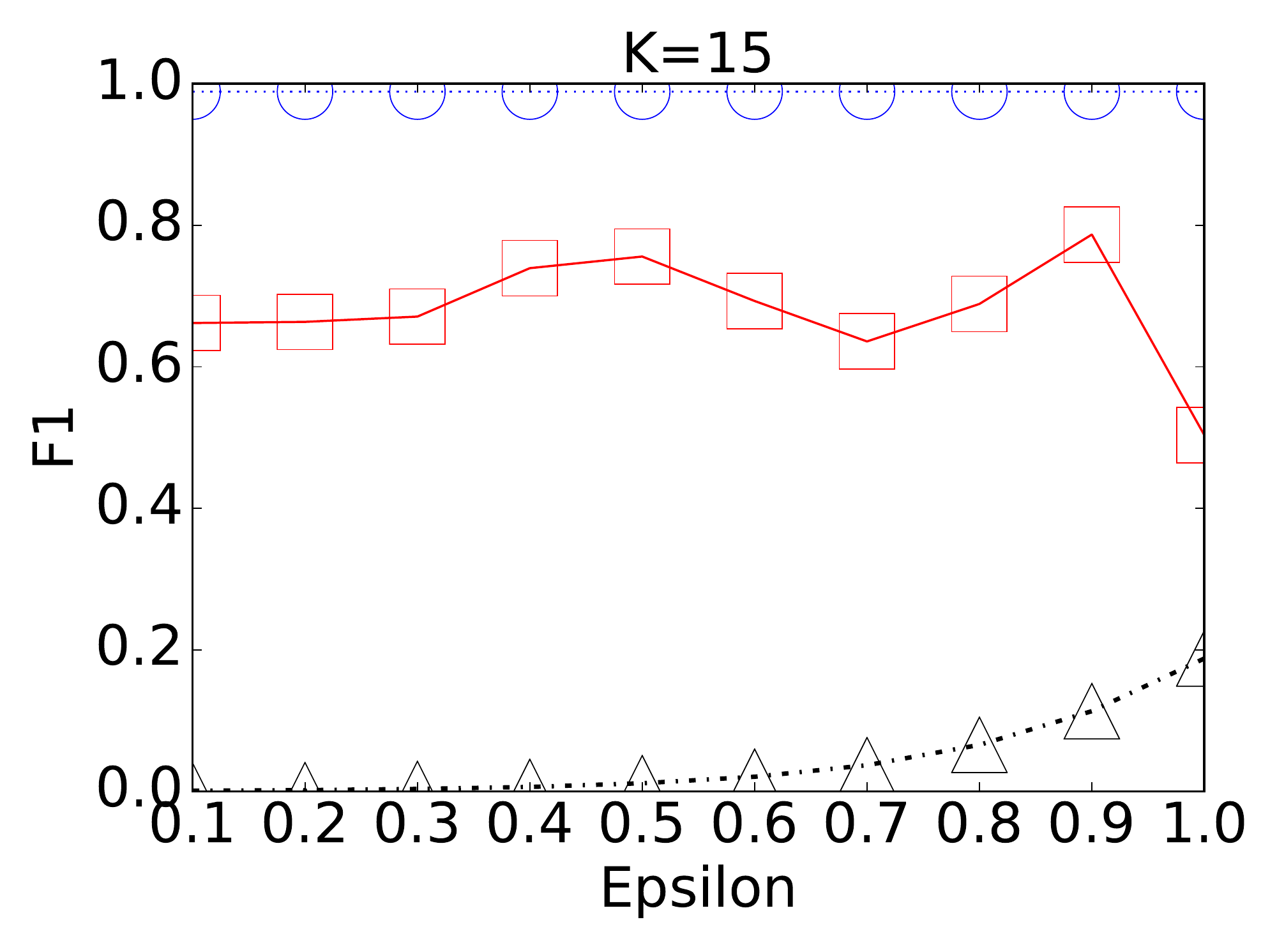}}\
\subfloat[BBC Sport]{\includegraphics[width = 0.3\linewidth]{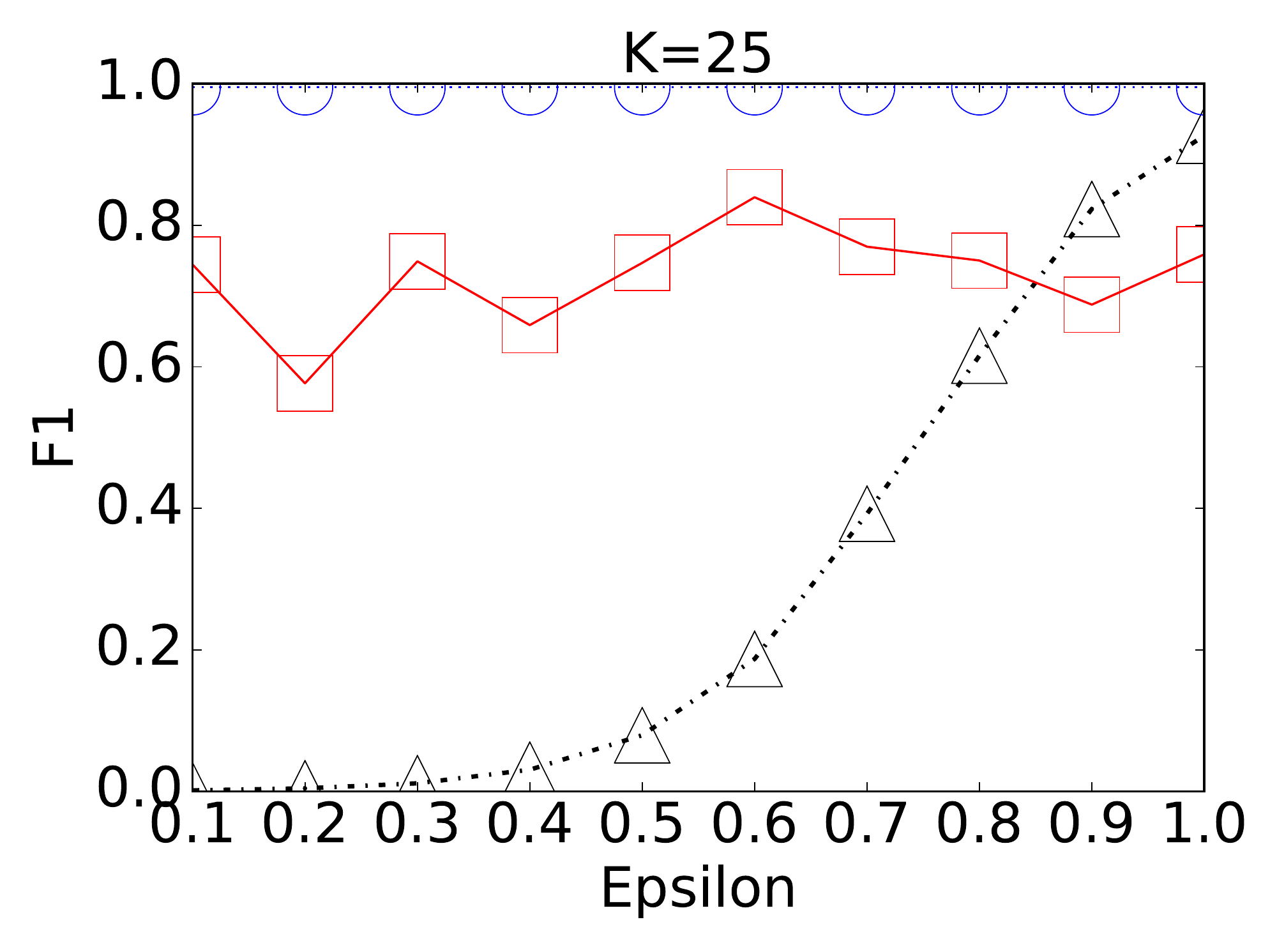}}\\
\subfloat[Opinosis]{\includegraphics[width = 0.3\linewidth]{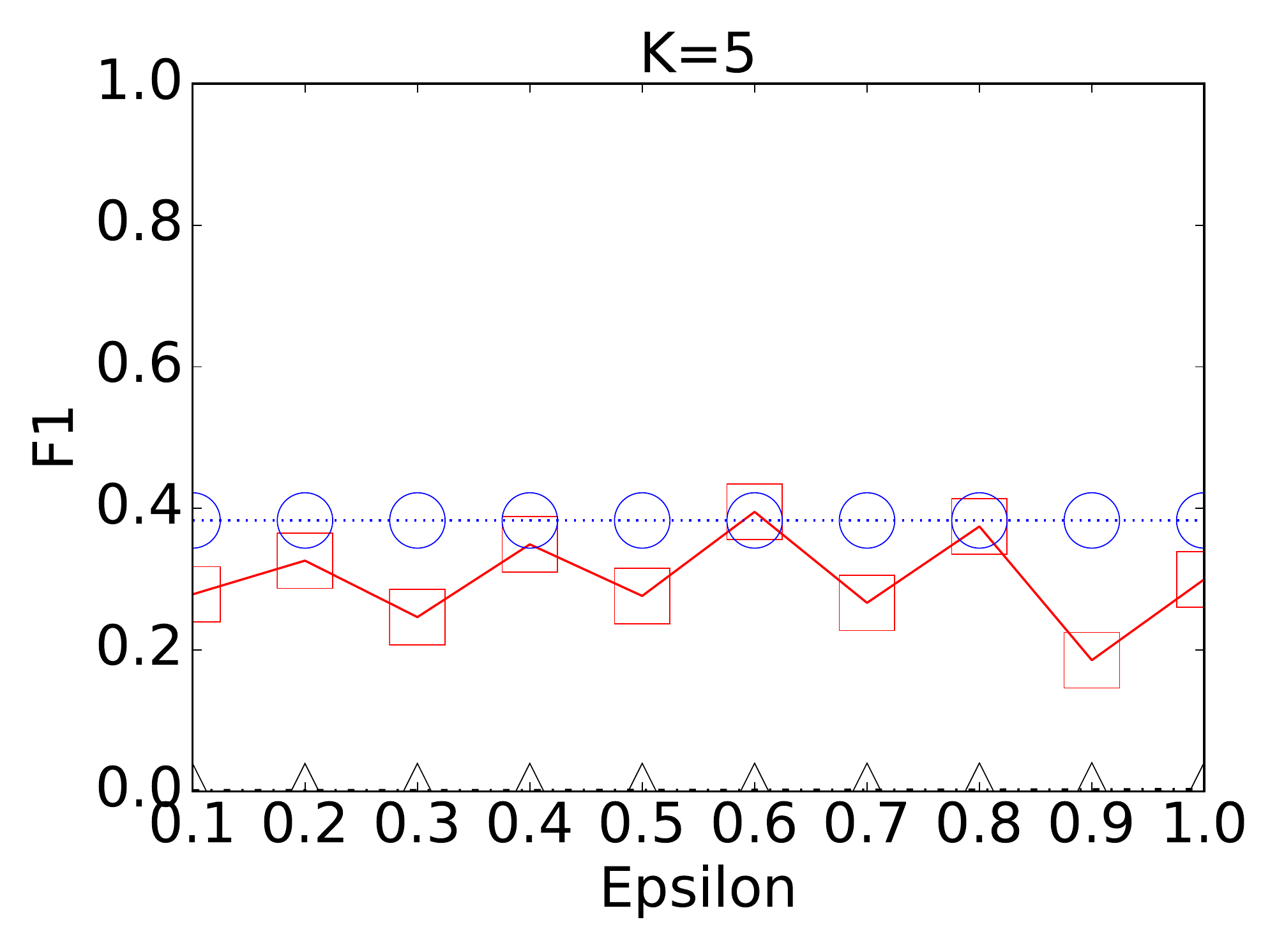}}\
\subfloat[Opinosis]{\includegraphics[width = 0.3\linewidth]{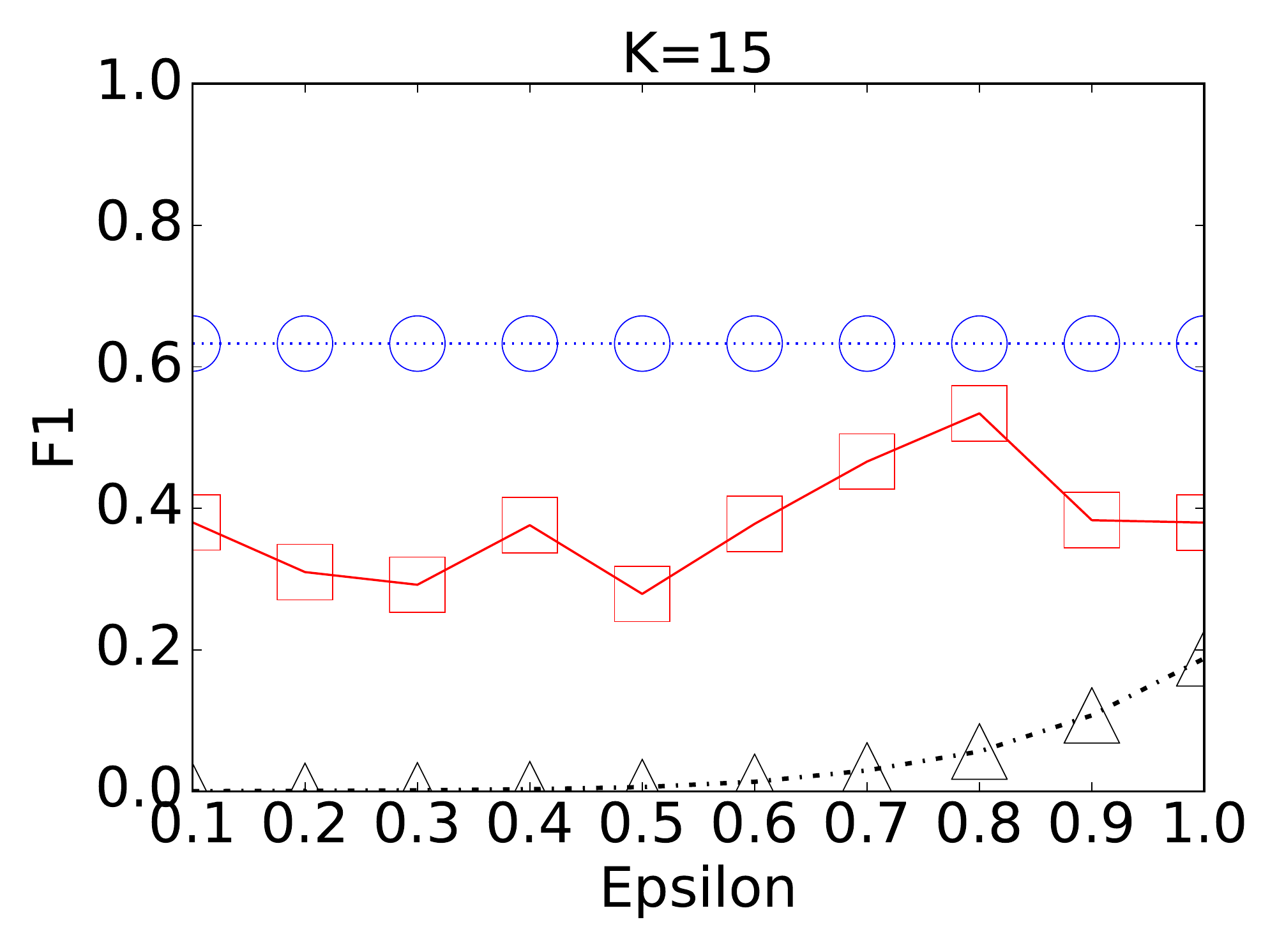}}\
\subfloat[Opinosis]{\includegraphics[width = 0.3\linewidth]{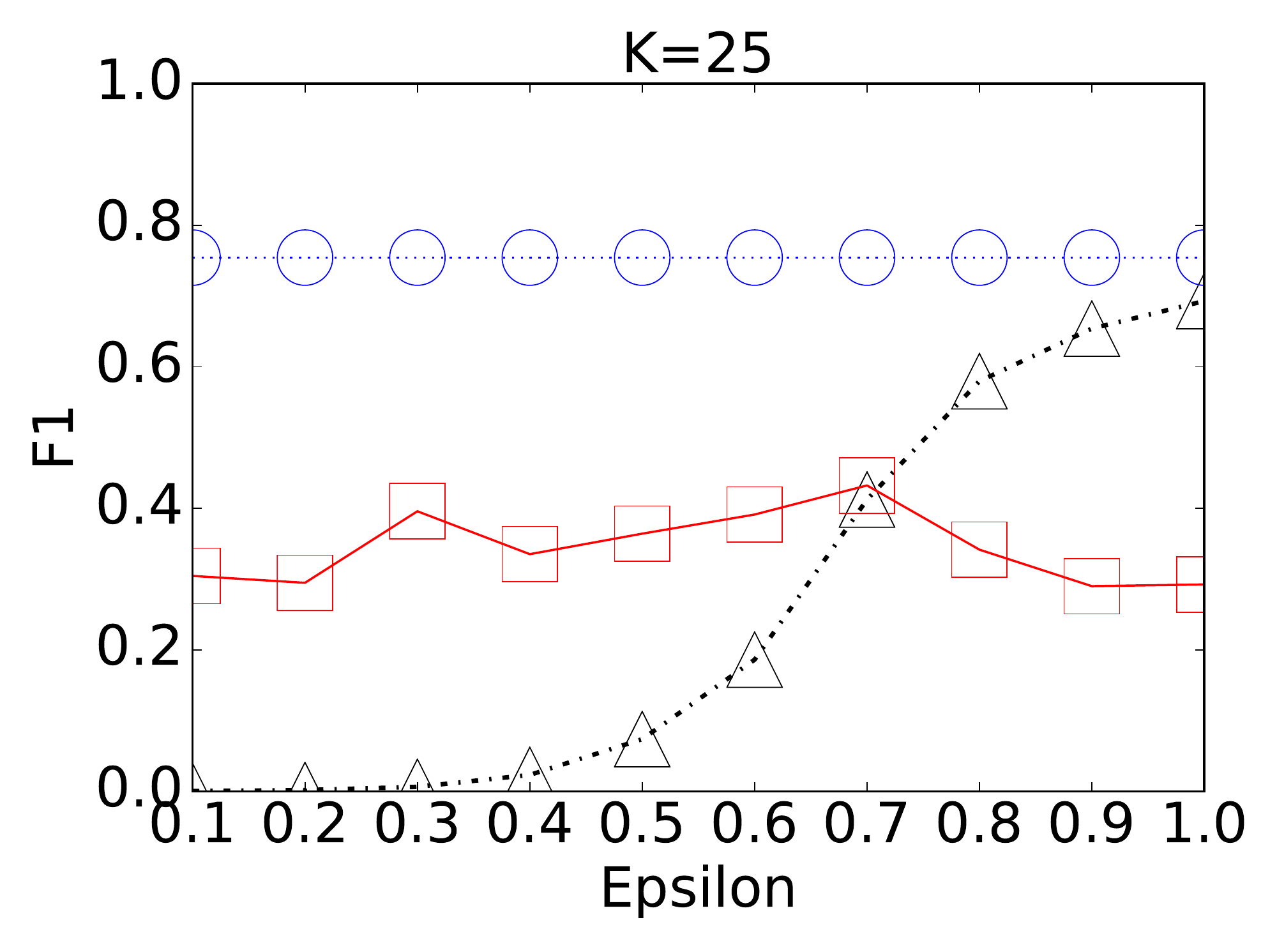}}\\
\subfloat[MSRP]{\includegraphics[width = 0.3\linewidth]{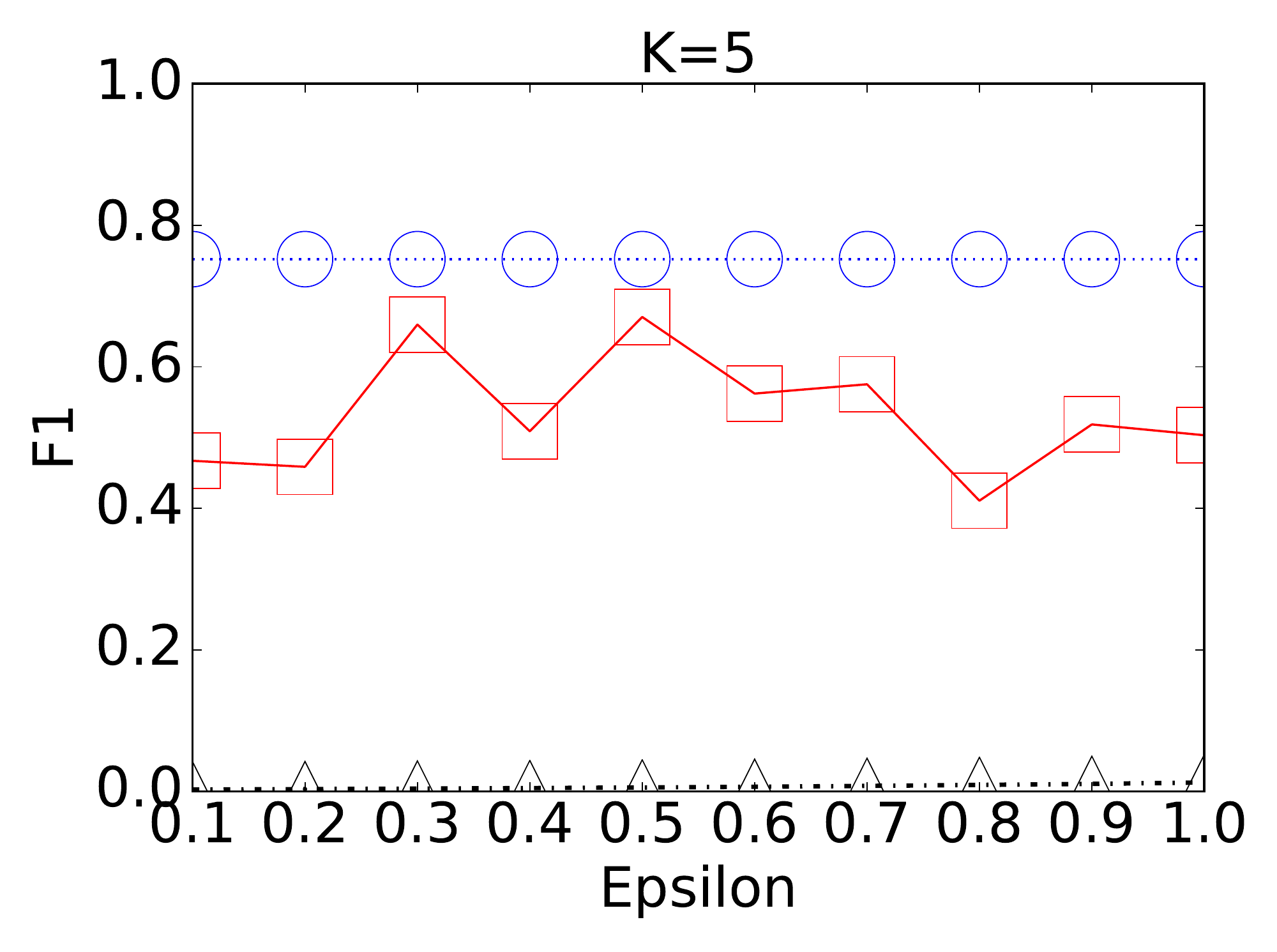}}\
\subfloat[MSRP]{\includegraphics[width = 0.3\linewidth]{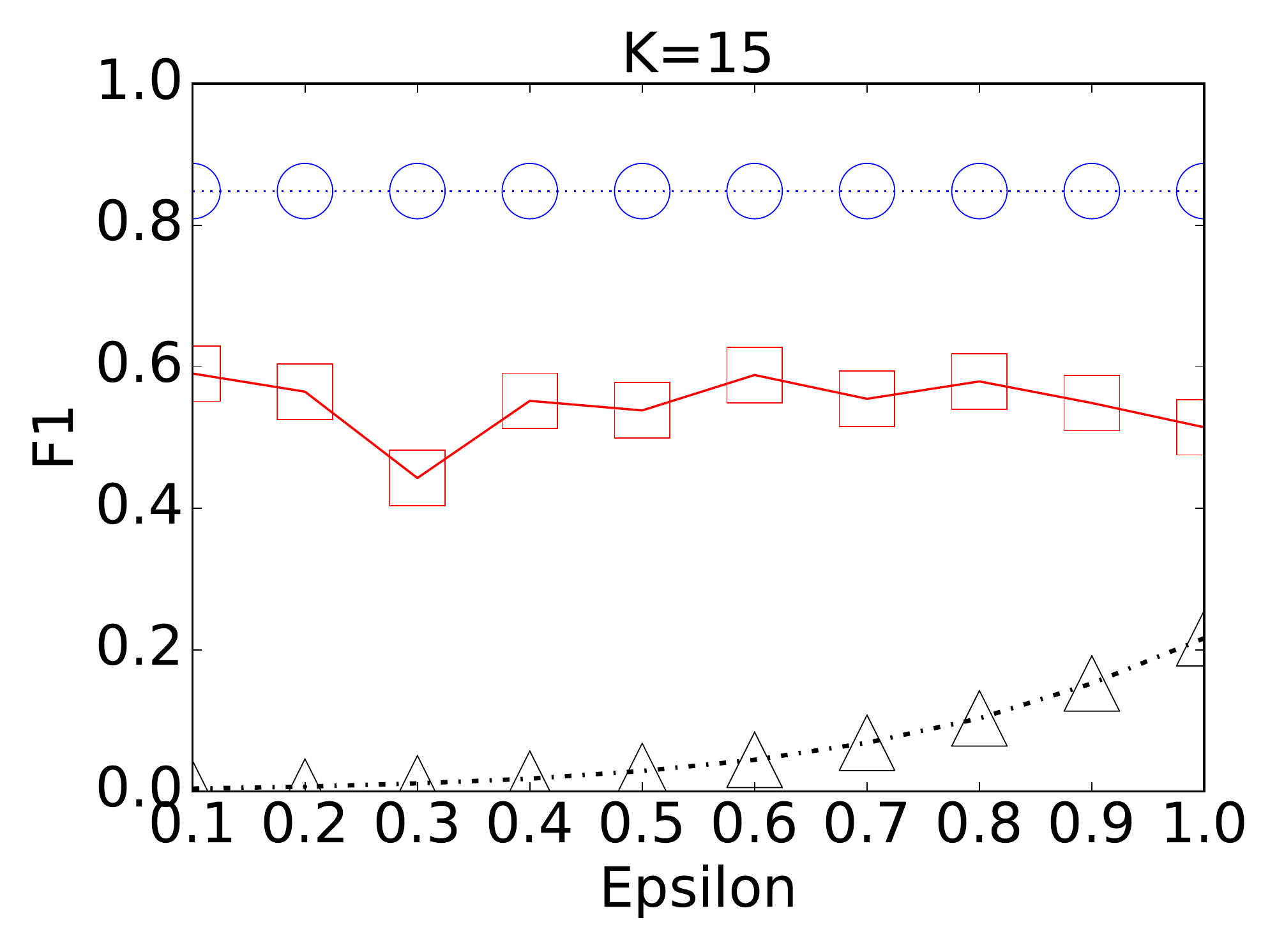}}\
\subfloat[MSRP]{\includegraphics[width = 0.3\linewidth]{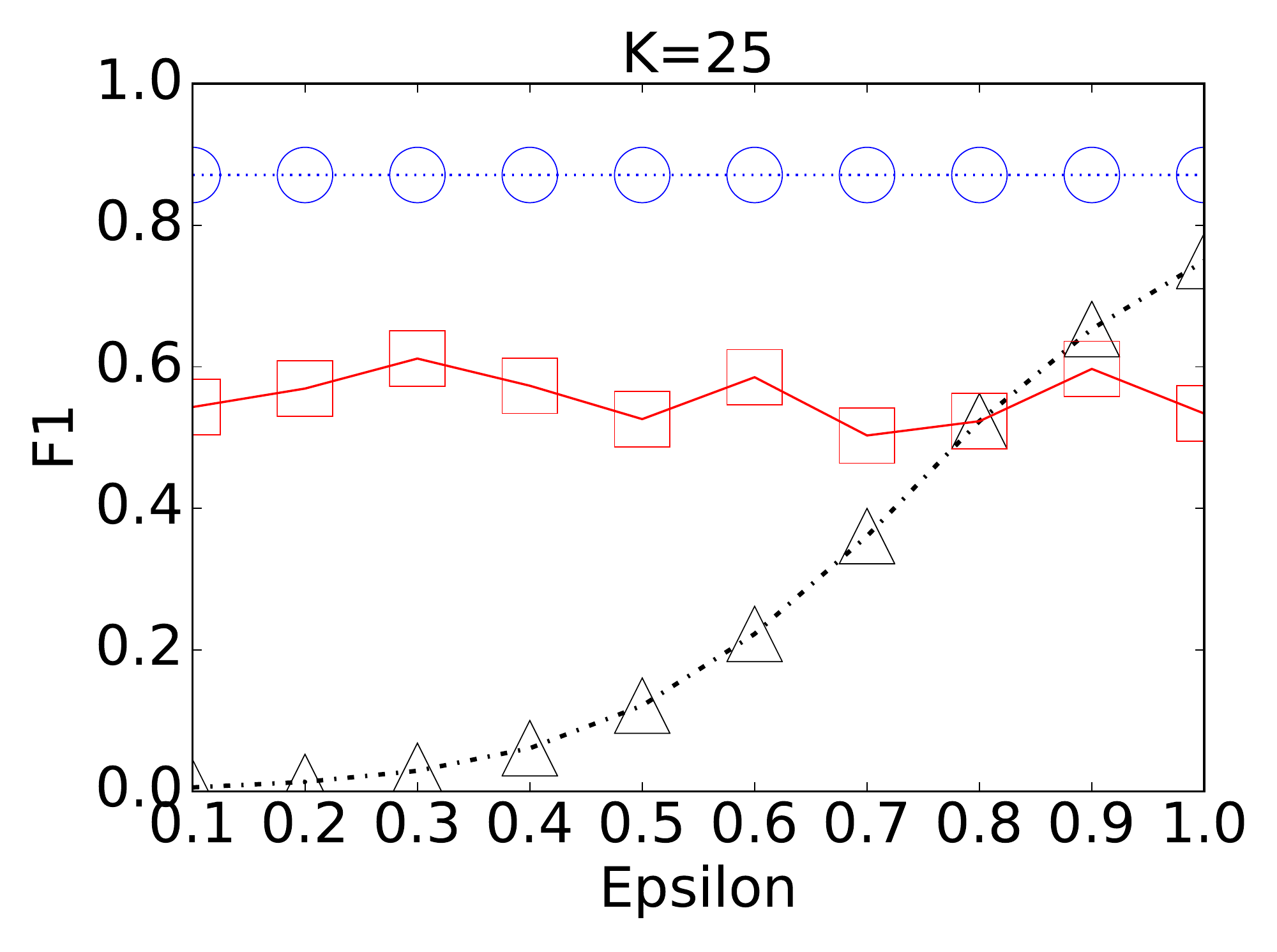}}\\
\caption{Varying $\epsilon$: F1}
\label{fig:varying-epsilon}
\end{figure*}

\begin{figure*}[htbp]
\centering
\includegraphics[scale=0.35]{fig_legend.pdf}\\
\subfloat[Alpine Dale]{\includegraphics[width = 0.3\linewidth]{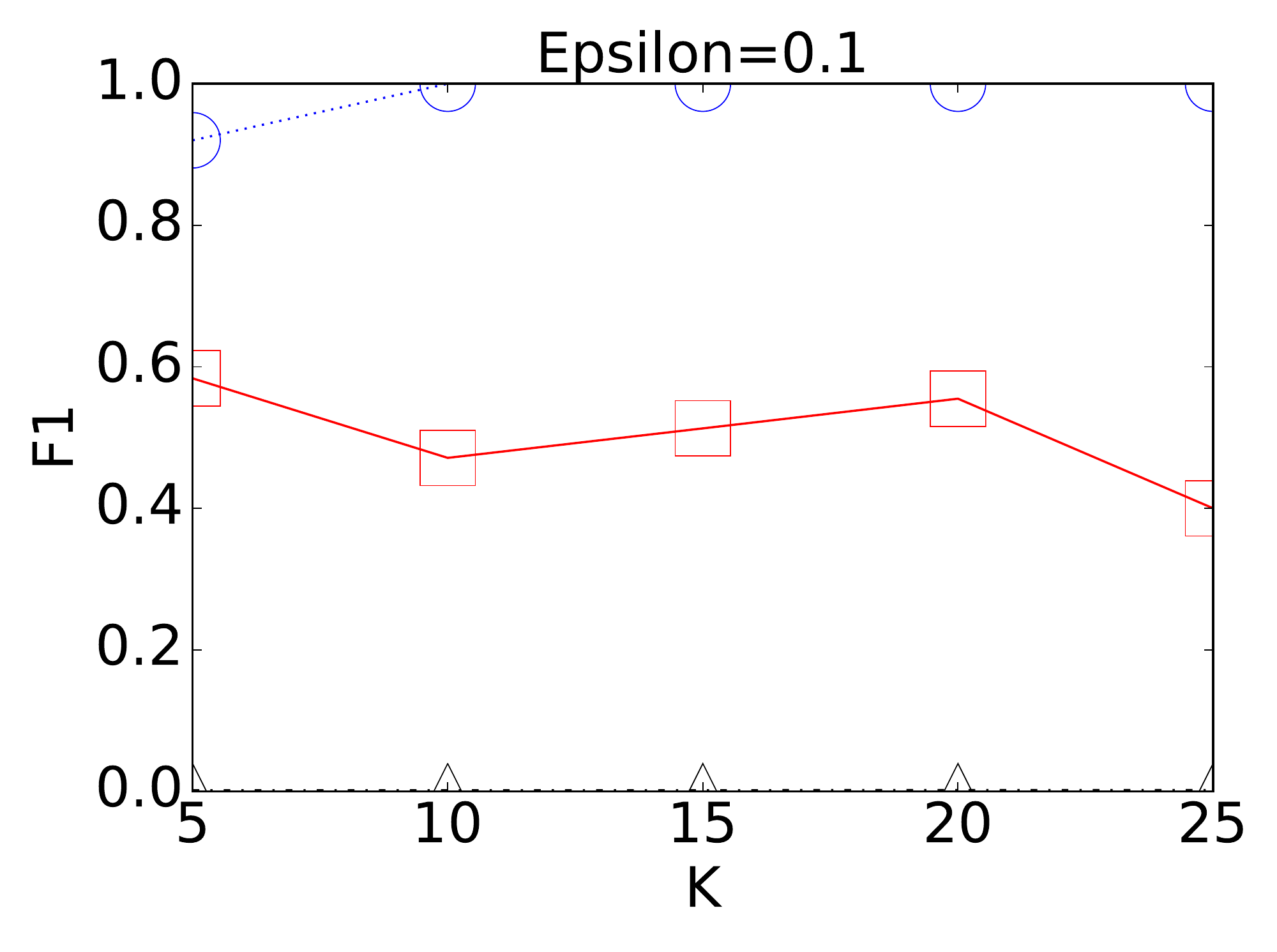}}\
\subfloat[Alpine Dale]{\includegraphics[width = 0.3\linewidth]{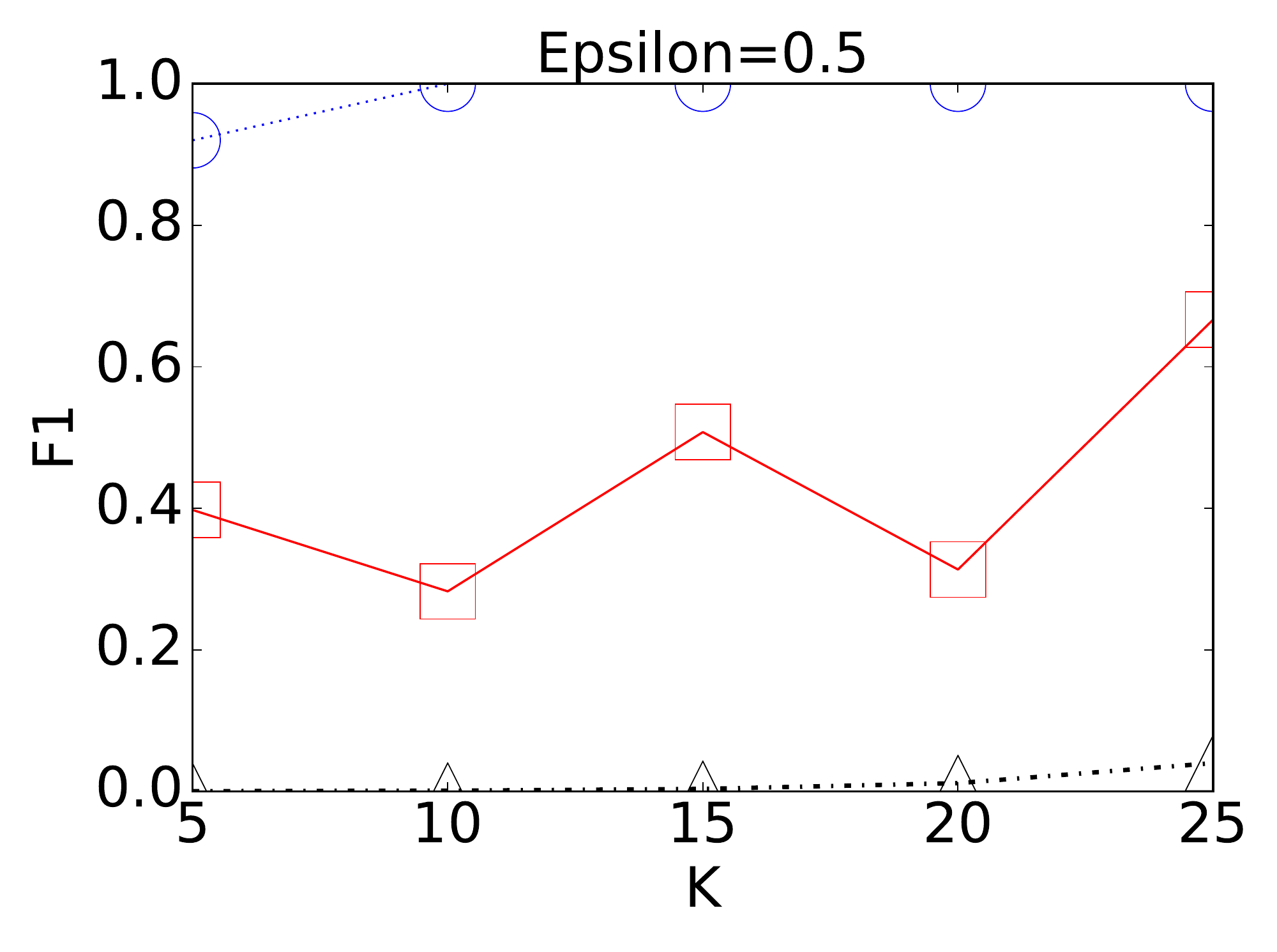}}\
\subfloat[Alpine Dale]{\includegraphics[width = 0.3\linewidth]{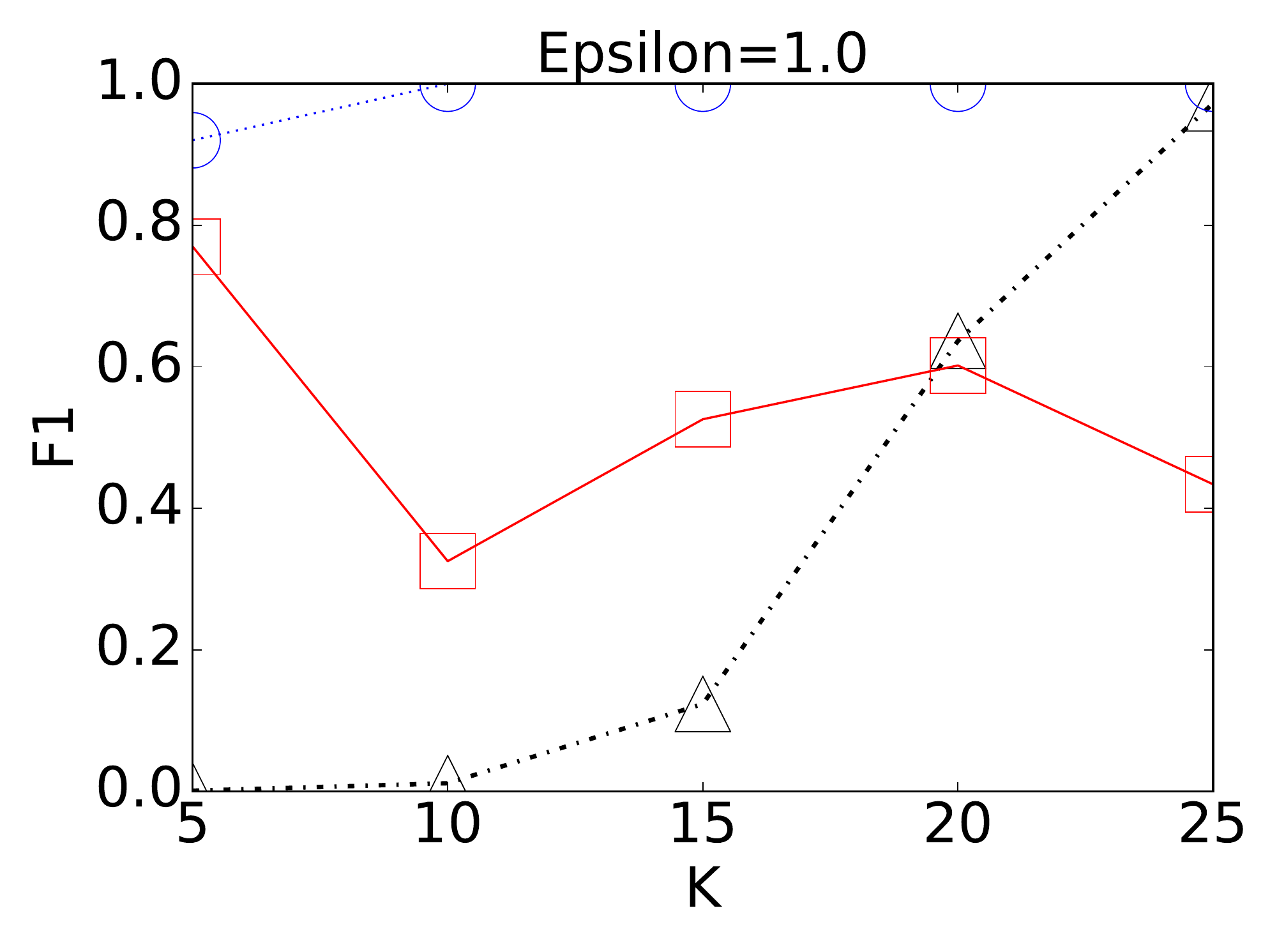}}\\
\subfloat[BBC Sport]{\includegraphics[width = 0.3\linewidth]{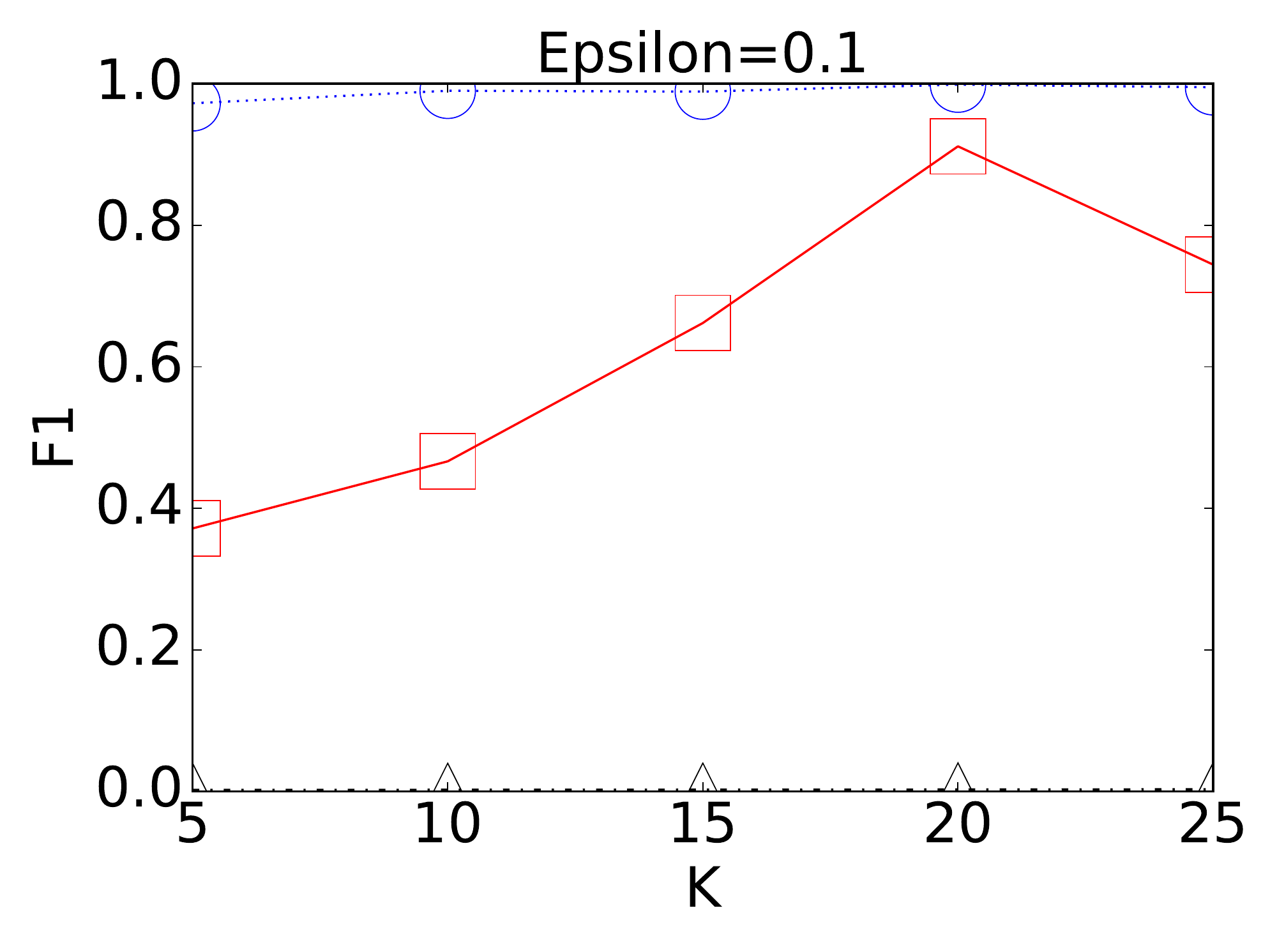}}\
\subfloat[BBC Sport]{\includegraphics[width = 0.3\linewidth]{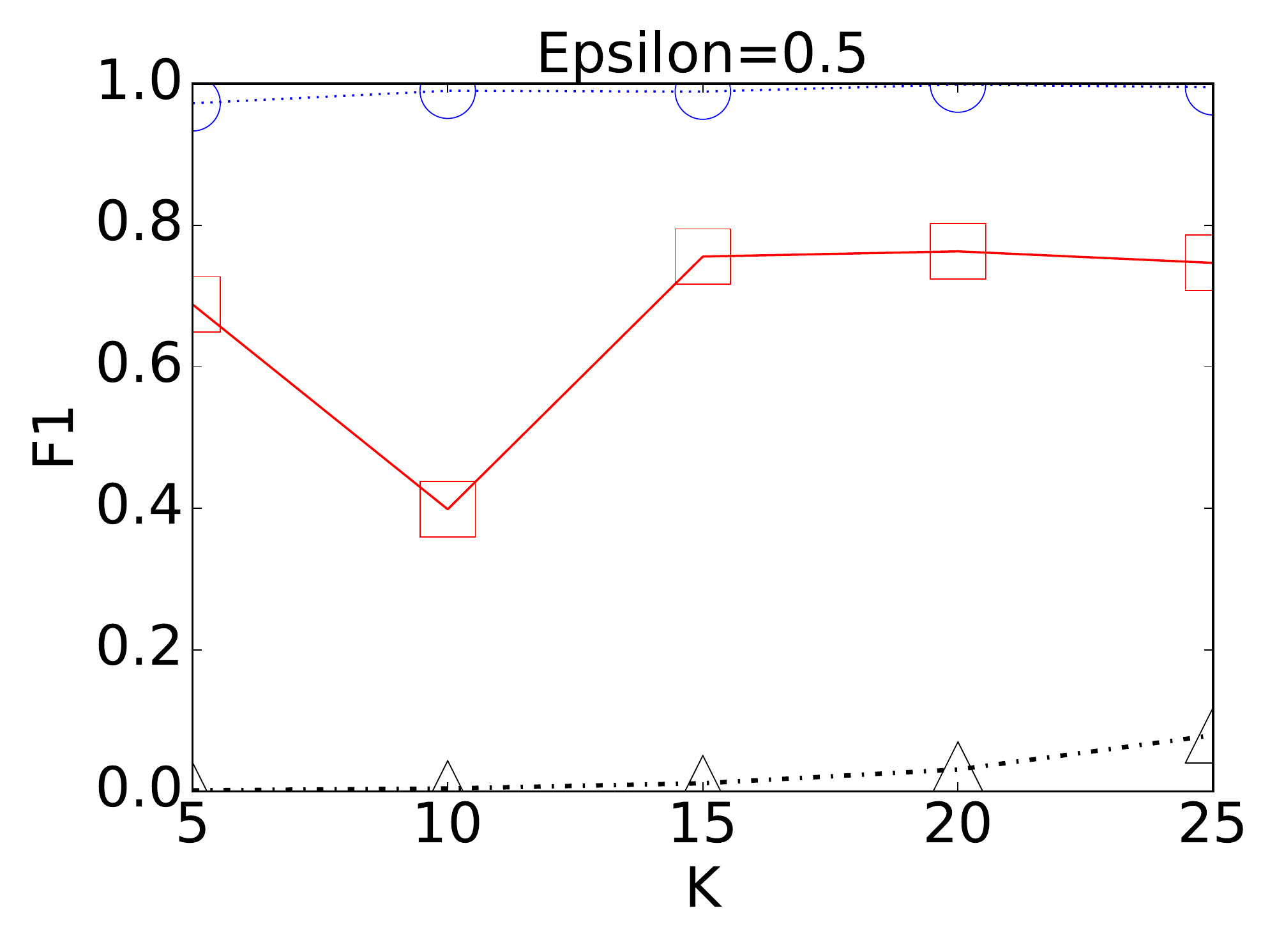}}\
\subfloat[BBC Sport]{\includegraphics[width = 0.3\linewidth]{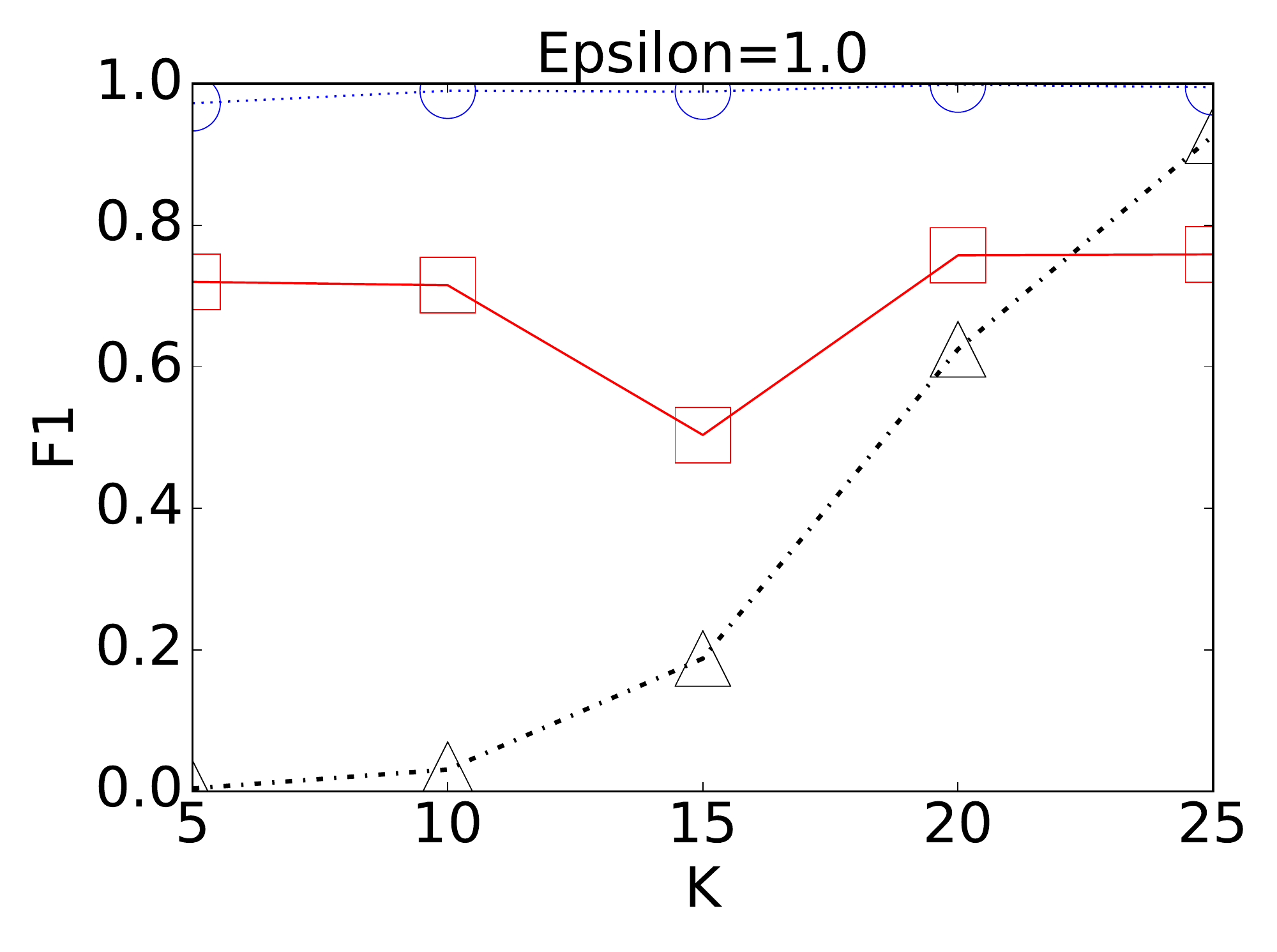}}\\
\subfloat[Opinosis]{\includegraphics[width = 0.3\linewidth]{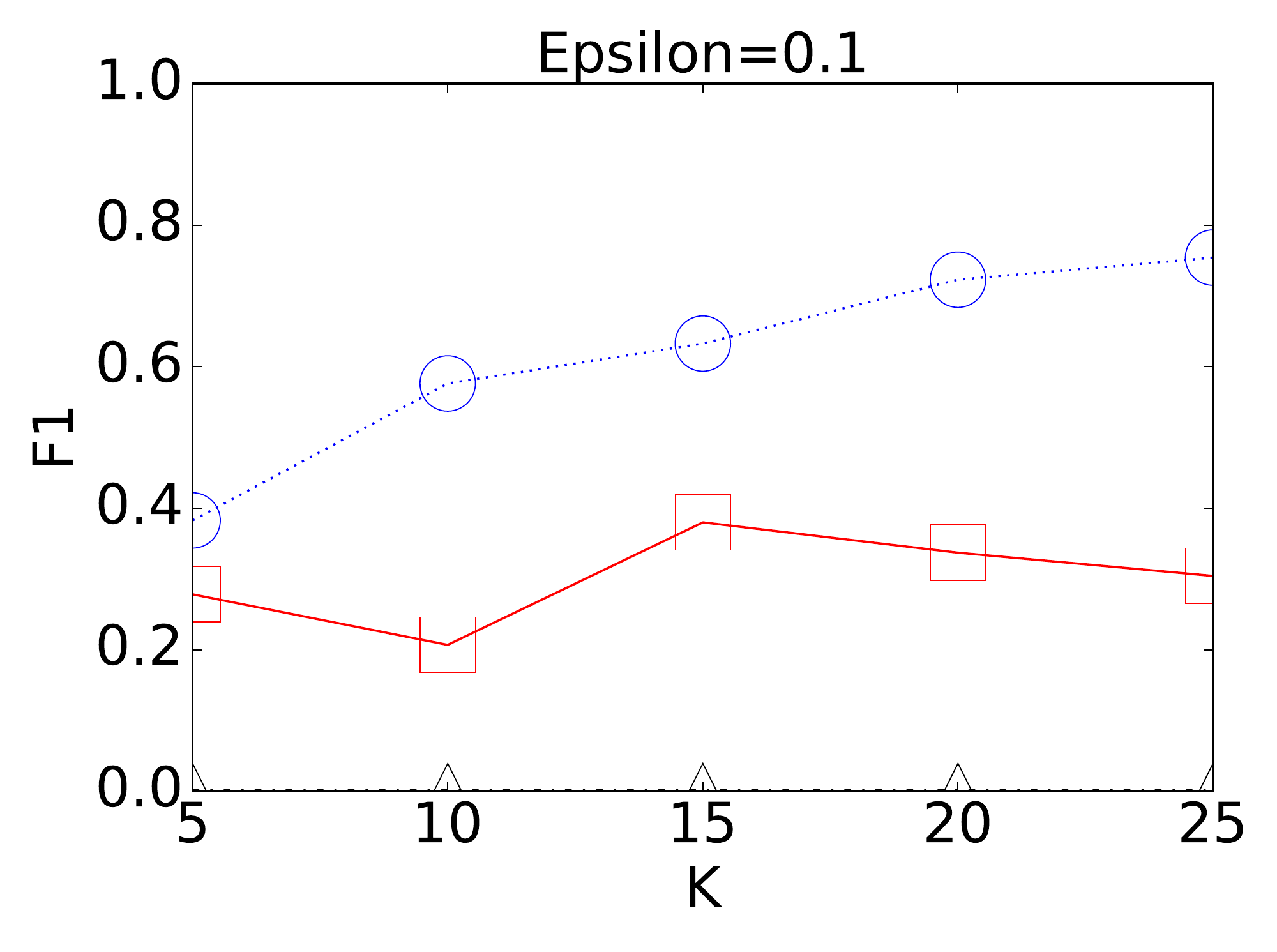}}\
\subfloat[Opinosis]{\includegraphics[width = 0.3\linewidth]{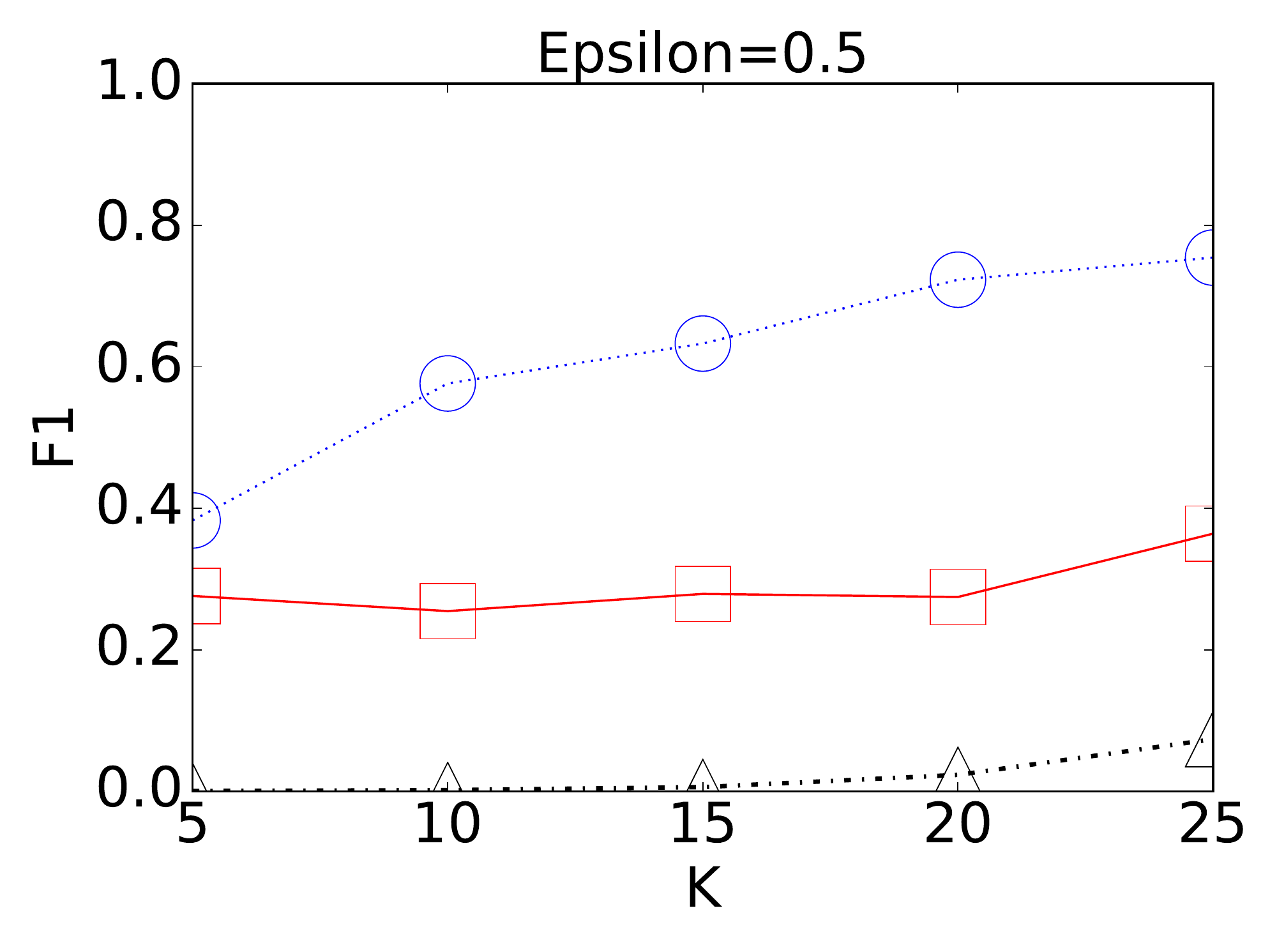}}\
\subfloat[Opinosis]{\includegraphics[width = 0.3\linewidth]{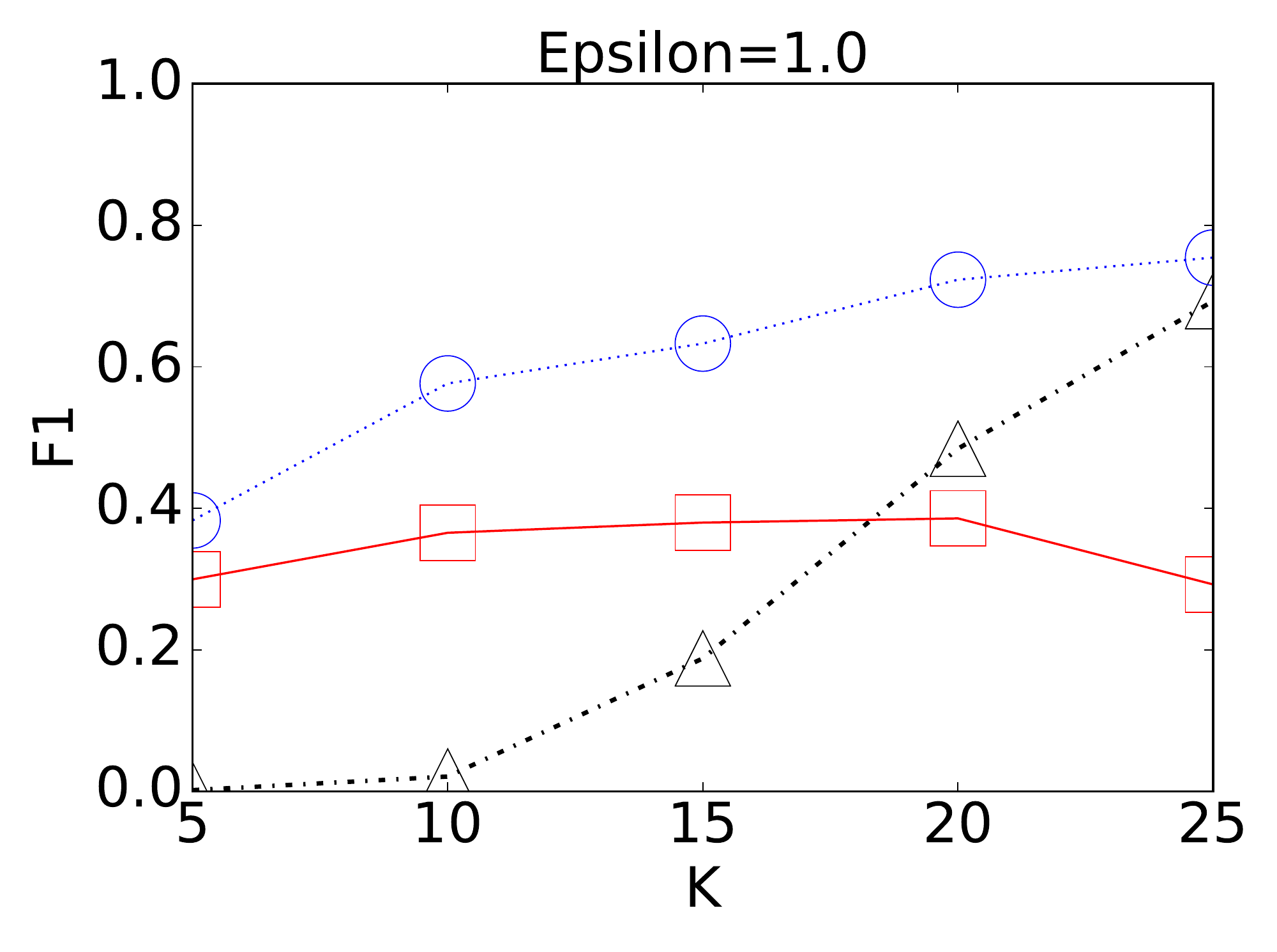}}\\
\subfloat[MSRP]{\includegraphics[width = 0.3\linewidth]{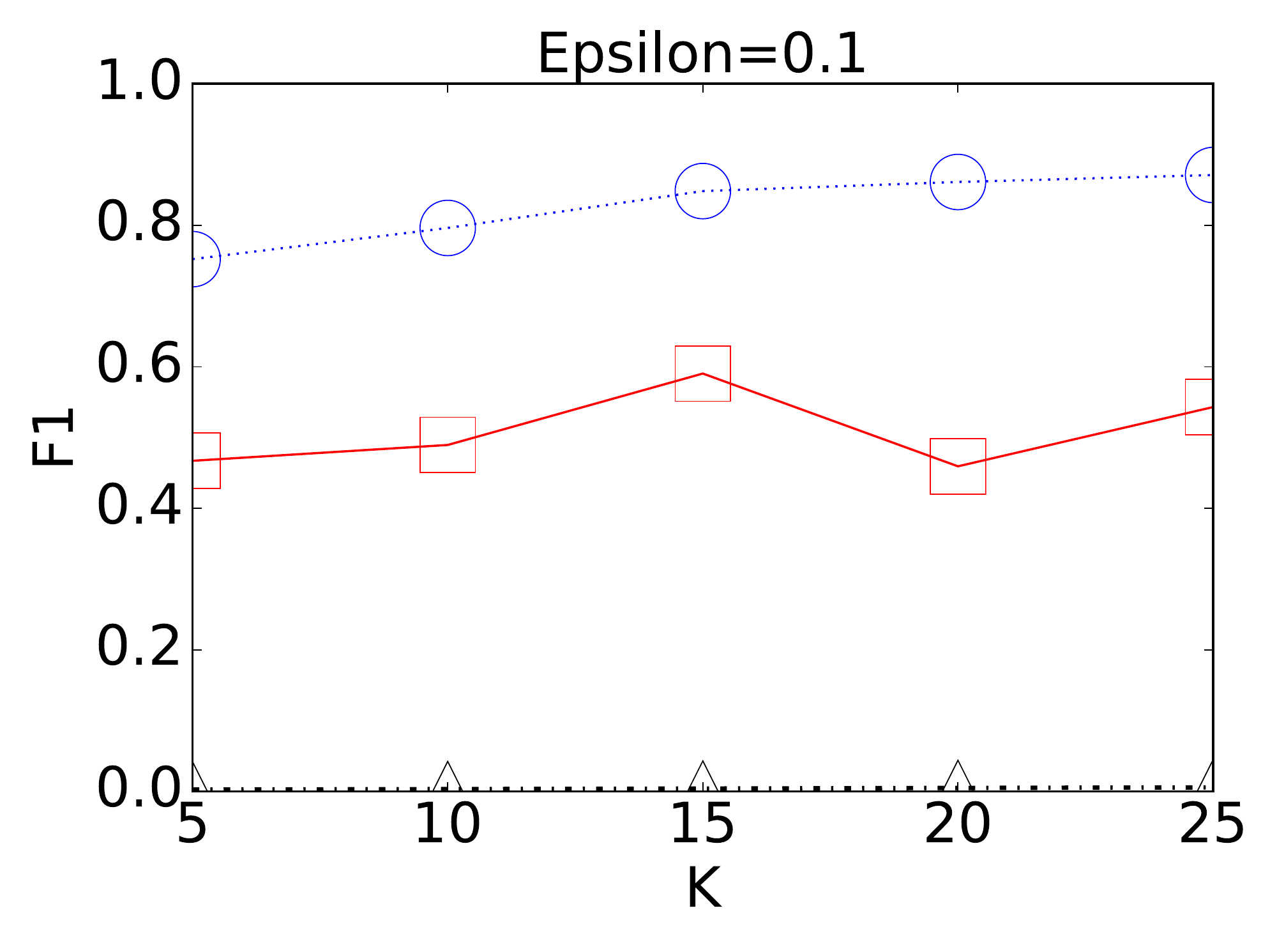}}\
\subfloat[MSRP]{\includegraphics[width = 0.3\linewidth]{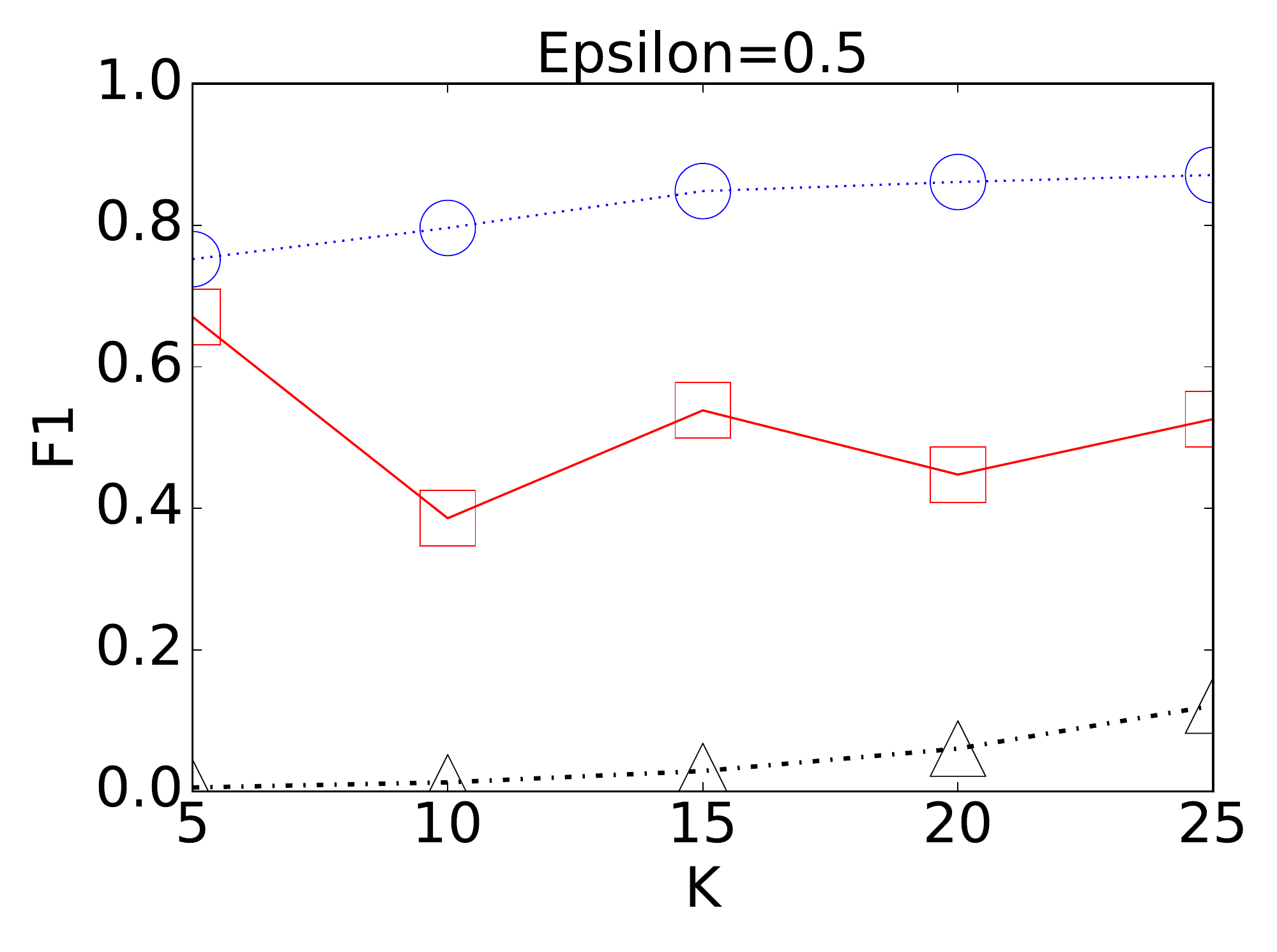}}\
\subfloat[MSRP]{\includegraphics[width = 0.3\linewidth]{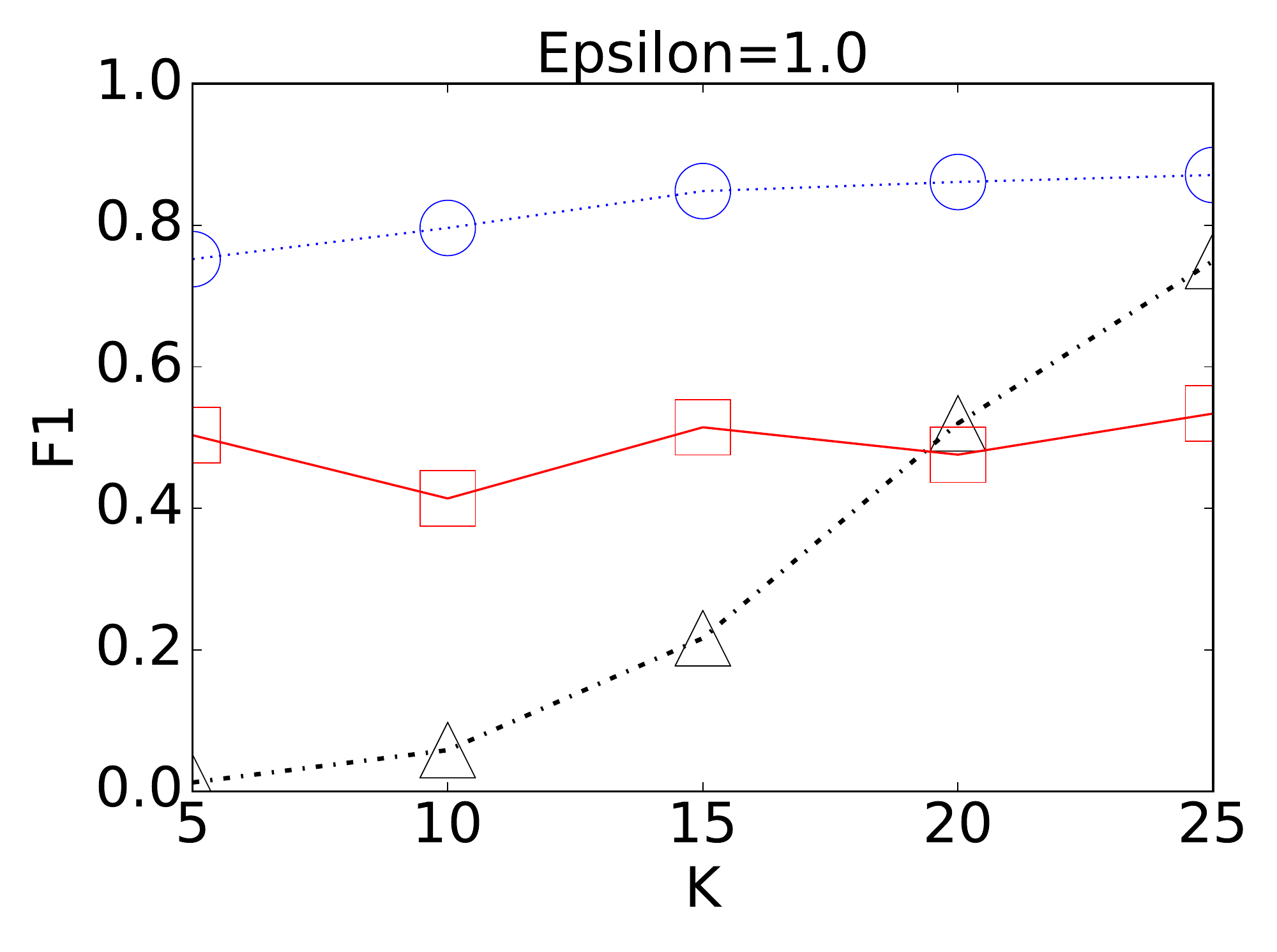}}\\
\caption{Varying $K$: F1}
\label{fig:varying-k}
\end{figure*}

\begin{table*}[htbp]
\scriptsize
\centering
\caption{Comparison of the MSE on Different Algorithms}
\label{tab:mse-comparison}
\renewcommand\arraystretch{1.3}
\begin{tabular}{ | p{1.7cm}<{\centering} | p{1cm}<{\centering} | p{1.1cm}<{\centering} | p{1.1cm}<{\centering} | p{1.1cm}<{\centering} | p{1.1cm}<{\centering} | p{1.1cm}<{\centering} | p{1.1cm}<{\centering} | p{1.1cm}<{\centering} | p{1.1cm}<{\centering} | p{1.1cm}<{\centering} |}
	\hline
    \multirow{2}{*}{\textbf{Parameters}}
    & $K$ & \multicolumn{3}{c|}{5} & \multicolumn{3}{c|}{15} & \multicolumn{3}{c|}{25} \\ \cline{2-11}	
    & $\epsilon$ & 0.01 & 0.1 & 1 & 0.01 & 0.1 & 1 & 0.01 & 0.1 & 1 \\ \hline
	\hline
    \multirow{3}{*}{\textbf{Alpine Dale}}
    & Baseline & 801.6931 & 8.003136 & 0.080065 & 88.86992 & 0.890652 & 0.008917 & 32.03135 & 0.320059 & 0.003222 \\ \cline{2-11}
    & $\texttt{PrivMin}$ & \textbf{0.000055} & 0.000060 & 0.000063 & \textbf{0.000024} & 0.000026 & 0.000026 & \textbf{0.000016} & 0.000017 & \textbf{0.000016} \\ \cline{2-11}
    & $\texttt{MH-JSC}$ & \multicolumn{3}{c|}{0.000117} & \multicolumn{3}{c|}{0.000031} & \multicolumn{3}{c|}{0.000022} \\ \hline
	\hline
    \multirow{3}{*}{\textbf{BBC Sport}}
    & Baseline & 798.8903 & 8.030287 & 0.080235 & 88.95124 & 0.89008 & 0.00896 & 31.94539 & 0.320657 & 0.003230 \\ \cline{2-11}	
    & $\texttt{PrivMin}$ & 0.000090 & \textbf{0.000088} & 0.000096 & 0.000065 & \textbf{0.000060} & 0.000062 & 0.000046 & \textbf{0.000044} & 0.000056 \\ \cline{2-11}	
    & $\texttt{MH-JSC}$ & \multicolumn{3}{c|}{0.000156} & \multicolumn{3}{c|}{0.000044} & \multicolumn{3}{c|}{0.000023} \\ \hline
    \hline
    \multirow{3}{*}{\textbf{Opinosis}}
    & Baseline & 800.8486 & 7.991925 & 0.080309 & 88.57464 & 0.888713 & 0.008927 & 32.03530 & 0.319903 & 0.003221 \\ \cline{2-11}
    & $\texttt{PrivMin}$ & 0.000100 & \textbf{0.000099} & 0.000106 & 0.000048 & 0.000053 & \textbf{0.000042} & 0.000037 & \textbf{0.000035} & \textbf{0.000035} \\ \cline{2-11}
    & $\texttt{MH-JSC}$ & \multicolumn{3}{c|}{0.000120} & \multicolumn{3}{c|}{0.000040} & \multicolumn{3}{c|}{0.000025} \\ \hline
    \hline
    \multirow{3}{*}{\textbf{MSRP}}
    & Baseline & 798.3747 & 7.996413 & 0.080062 & 89.04559 & 0.887925 & 0.008916 & 31.98979 & 0.320710 & 0.003216 \\ \cline{2-11}
    & $\texttt{PrivMin}$ & 0.000135 & 0.000114 & \textbf{0.000108} & \textbf{0.000099} & 0.000124 & 0.000101 & 0.000084 & 0.000089 & \textbf{0.000073} \\ \cline{2-11}
    & $\texttt{MH-JSC}$ & \multicolumn{3}{c|}{0.000095} & \multicolumn{3}{c|}{0.000031} & \multicolumn{3}{c|}{0.000018} \\ \hline
\end{tabular}
\end{table*}

\subsubsection{Impact of Hash Function Number}
\label{sec-impact-of-hash-fuction-number}

For the \emph{$F1$ Score},
we fix $\epsilon=0.1,0.5,1.0$ and
vary the size $K$ of a single MinHash signature,
to study its impact on the utility of each algorithm.
The results are shown in Fig.~\ref{fig:varying-k}.
As expected,
for the Baseline algorithm its utility measures increase
when $K$ increases.
We also observe that the \texttt{PrivMin} algorithm
is not clearly affected by the changing $K$.

For the \emph{Mean Squared Error} (MSE),
Table~\ref{tab:mse-comparison} shows that
the \texttt{PrivMin} algorithm generally outperforms the
\emph{Baseline} algorithm with the changing $K$.
And in some conditions,
it also maintains a better utility compared with
the \texttt{MH-JSC}.

\subsubsection{Summary and Recommendations}

Remarkably,
although the \emph{Baseline} algorithm can achieve
$\epsilon$-differential privacy
by adopting the \emph{Laplace} mechanism,
it can hardly maintain an acceptable \emph{$F1$ Score}
unless both $\epsilon$ and $K$ are large
(e.g.,
$\epsilon \geq 1.0$ and $K \geq 20$,
or $\epsilon \geq 0.8$ and $K \geq 25$ in experimental results
shown in Fig.~\ref{fig:varying-k} and Fig.~\ref{fig:varying-epsilon}).
This is because that
in the design of the \emph{Baseline} algorithm,
the \texttt{MH-JSC} is directly considered as a ``black box'',
and the noise adding is applied on this box's output
without exploiting the minimum hash value computation process
inside the box.
In contrast,
the intuition behind the \texttt{PrivMin} algorithm
is that introducing the \emph{Exponential} noise to the
minimum hash value computation process
of MinHashing phase in a privately randomized way,
which helps the proposed algorithm achieve both
$\epsilon$-differential privacy and acceptable utility.
Based on the above empirical results in terms of utility metrics,
when using the proposed \texttt{PrivMin} algorithm,
we recommend that $\epsilon \leq 1.0$ and $K \leq 20$ or,
$\epsilon \leq 0.8$ and $K \leq 25$.

\section{Conclusions}\label{sec-conclusion}

\emph{Jaccard Similarity Computation} is an essential
process which has been widely used in many real-world applications
such as recommendation and plagiarism detection.
However,
its potential privacy leakage
is an emerging issue that needs to be addressed.
Current research pay little attention on the
\emph{MinHash-based Jaccard Similariy Computation}
(\texttt{MH-JSC})
for designing a differentially private algorithm.
This paper studies the \texttt{MH-JSC}
under the relaxed and the strict differential privacy
with the following contributions:
\begin{itemize}
\item
We first provide a relaxed definition of $\epsilon-$DPSO
that extends the differential privacy into set operations.
It is found that the \texttt{MH-JSC}
satisfies the \emph{Conditional $\epsilon-DPSO$} naturally.
Relevant theorem and detailed proof of privacy analysis
are provided in Section~\ref{sec-dpso}.

\item
Based on the above analysis,
we then proposed the \texttt{PrivMin} algorithm
in Section~\ref{sec-private-jaccard-similarity-computation}
to achieve the differential privacy.
The proposed algorithm consists of two private operations,
the \emph{Private MinHashing Value Generation}
that applies the naive \emph{Exponential} mechanism
for the MinHashing phase,
and the \emph{Randomized MinHashing Steps Selection} which
takes the advantages of the \emph{Randomized Response} technique.
\end{itemize}

These contributions constitute a practical solution
to the differentially private Jaccard similarity computation
with less utility loss.
Our theoretical and experimental analysis
in Section~\ref{sec-analysis} and~\ref{sec-experiment}
show that the proposed \texttt{PrivMin} algorithm could reserve
acceptable utility.

\newpage

\bibliographystyle{plainnat}
\bibliography{PrivMin-arXiv}

\end{document}